%% file: main.tex
\documentclass[runningheads]{llncs}
\usepackage[T1]{fontenc}
\usepackage{graphicx}
\usepackage{hyperref}

\usepackage{subcaption}
\usepackage{fancybox}
    {\begin{Sbox}\begin{minipage}{#1}}%
	{\end{minipage}\end{Sbox}\fbox{\TheSbox}}

\newenvironment{framed}%
    {\begin{Sbox}\begin{minipage}{.98\textwidth}}%
	{\end{minipage}\end{Sbox}\fbox{\TheSbox}}

\usepackage{tabularx}
\usepackage{multirow}

\usepackage{tikz}
\usepackage{pgfplots}

\definecolor{oxgreen}{RGB}{105,145,59}
\definecolor{oxred}{RGB}{190,15,52}
\definecolor{oxlighterblue}{HTML}{1a385a}

\usepackage{booktabs}   

\newtheorem{assumption}{Assumption}

\usepackage{proof}
\usepackage{amsfonts,amssymb,mathtools,amsmath}
\usepackage{stmaryrd}
\usepackage{thm-restate}
\usepackage{enumitem}
\usepackage{wrapfig}
\usepackage[misc]{ifsym}

\usepackage[capitalize]{cleveref}
\crefname{assumption}{Assumption}{Assumptions}
\crefname{problem}{Problem}{Problems}

\newif\ifdraft
\draftfalse

\usepackage{import}
\import{.}{comment.tex}

\input{macros}

\def\orcidID#1{\smash{\href{http://orcid.org/#1}{\protect\raisebox{-1.25pt}{\protect\includegraphics{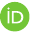}}}}}

\begin{document}
%
\title{Fast and Correct Gradient-Based Optimisation for Probabilistic Programming via Smoothing}
\titlerunning{Fast and Correct Optimisation for Probabilistic Programming via Smoothing}

%
\author{Basim Khajwal\inst{1} \and
C.-H. Luke Ong\inst{1,2}\orcidID{0000-0001-7509-680X} \and
Dominik Wagner\inst{1}$^{\text{(\Letter)}}$ \orcidID{0000-0002-2807-8462}}
\institute{University of Oxford \and
NTU Singapore}
%
\maketitle              
\begin{abstract}
  We study the foundations of variational inference, which frames posterior inference as an optimisation problem, for probabilistic programming.
  The dominant approach for optimisation in practice is stochastic gradient descent. In particular, a variant using the so-called reparameterisation gradient estimator exhibits fast convergence in a traditional statistics setting.
  Unfortunately, discontinuities, which are readily expressible in programming languages, can compromise the correctness of this approach.
  We consider a simple (higher-order, probabilistic) programming language with conditionals,
  and we endow our language with both a measurable and a \emph{smoothed} (approximate) value semantics.
  We present type systems which establish technical pre-conditions. Thus we can prove stochastic gradient descent with the reparameterisation gradient estimator to be correct when applied to the smoothed problem. Besides, we can solve the original problem up to any error tolerance by choosing an accuracy coefficient suitably.
  Empirically we demonstrate that our approach has a similar convergence as a key competitor, but is simpler, faster, and attains orders of magnitude reduction in work-normalised variance.

\end{abstract}

\section{Introduction}
\label{sec:intro}

\input{sections/intro}


\section{A Simple Programming Language}
\label{sec:PL}



\input{sections/PL}

%
%

\section{Smoothed Denotational Value Semantics}
\label{sec:smoothing}
\input{sections/smooth}


\section{Correctness of SGD for Smoothed Problem and Unbiasedness of the Reparameterisation Gradient}
\label{sec:sgd}
\input{sections/sgd}




\section{Uniform Convergence}
\label{sec:conv}
\input{sections/conv}



\section{Related Work}
\label{sec:related}

\input{sections/related}

\section{Empirical Evaluation}
\label{sec:empirical}

\input{sections/experiments}

\section{Conclusion and Future Directions}
\label{sec:conclusion}
We have discussed a simple probabilistic programming language to formalise an optimisation problem arising e.g.\ in variational inference for probabilistic programming. We have endowed our language with a denotational (measurable) value semantics and a smoothed approximation of potentially discontinuous programs, which is parameterised by an accuracy coefficient.
We have proposed type systems to guarantee pleasing properties in the context of the optimisation problem:
For a fixed accuracy coefficient, stochastic gradient descent converges to stationary points even with the reparameterisation gradient (which is \emph{unbiased}).
Besides, the smoothed objective function converges uniformly to the true objective as the accuracy is improved.

Our type systems can be used to \emph{independently} check these two properties to obtain partial theoretical guarantees even if one of the systems suffers from incompleteness.
We also stress that SGD and the smoothed unbiased gradient estimator can even be applied to programs which are \emph{not} typable.


Experiments with our prototype implementation confirm the benefits of reduced variance and unbiasedness.
Compared to the unbiased correction of the reparameterised gradient estimator due to \cite{LYY18},
our estimator has a similar convergence, but is simpler, faster, and attains orders of magnitude (2 to 3,000 x) reduction in work-normalised variance.

\paragraph{Future Directions.}
A natural avenue for future research is to make the language and type systems more complete, i.e.\ to support more well-behaved programs, in particular programs involving recursion.



Furthermore, the choice of accuracy coefficients leaves room for further investigations. We anticipate it could be fruitful not to fix an accuracy coefficient upfront but to gradually enhance it \emph{during} the optimisation either via a pre-determined schedule (dependent on structural properties of the program), or adaptively.


%


\bibliographystyle{splncs04}
\bibliography{lit}

\newpage

\appendix

\section{Supplementary Materials for \cref{sec:PL}}
\label{app:PL}
\input{appendix/app-PL}

\section{Supplementary Materials for \cref{sec:smoothing}}
\label{app:smooth}
\input{appendix/app-smooth}

\section{Supplementary Materials for \cref{sec:sgd}}
\input{appendix/app-sgd}



\section{Supplementary Materials for \cref{sec:conv}}
\label{app:conv}
\input{appendix/app-conv}

\section{Supplementary Materials for \cref{sec:empirical}}
\input{appendix/app-experiments}


\end{document}

%% file: comment.tex
\usepackage{savesym}
\savesymbol{comment}

\ifdraft
\usepackage[draft]{commenting}
\else
\usepackage[nompar]{commenting}
\fi

\declareauthor{lo}{LO}{teal}
\authorcommand{lo}{comment}

\declareauthor{dw}{DW}{blue}
\authorcommand{dw}{comment}

\declareauthor{bk}{BK}{green}
\authorcommand{bk}{comment}
\setdefaultauthor{dw}

\declareauthor{bk}{BK}{green}
\authorcommand{bk}{comment}

\makeatletter
\renewcommand{\comm@todo@mpar}[1]{}
\makeatother

\def\divider{%
  \leavevmode\leaders\hrule height 0.6ex depth \dimexpr0.4pt-0.6ex\hfill%
  \kern0pt%
}

\renewcommand{\OpenCommentBraket}{$\lceil$}
\renewcommand{\ClosedCommentBraket}{$\rfloor$}

%% file: macros.tex
\newcommand\defn[1]{\emph{#1}}
\newcommand\defeq{\coloneqq}

\newcommand\Real{\mathbb R}
\newcommand\Realp{\Real_{>0}}

\newcommand\compi\rightsquigarrow

\newcommand\csfc\rightsquigarrow

\newcommand\sd{T}

\newcommand\nat{\mathbb N}

\newcommand\sample[1]{\mathbf{sample}\,#1}
\newcommand\score{\mathbf{score}}
\newcommand\ifc[3]{\mathbf{if }\,#1 <0\,\mathbf{ then }\,#2\,\mathbf{ else }\,#3}

\newcommand\Set{\mathbf{Set}}

\newcommand\Fr{\mathbf{Fr}}
\newcommand\VFr{\mathbf{VectFr}}
\newcommand\qbs{\mathbf{QBS}}
\newcommand\PCFReal{R}
\newcommand\PCFlog{\underline\log}
\newcommand\PCFexp{\underline\exp}
\newcommand\PCFplus{\underline+}
\newcommand\PCFmul{\underline\cdot}
\newcommand\PCFmin{\underline -}
\newcommand\PCFcirc{\underline\circ}

\newcommand\PCFinv{\underline{{}^{-1}}}

\newcommand\PCFpower{\underline{{}^a}}
\newcommand\PCFtimes{\underline{a\,\cdot}\,}

\newcommand\genj{\vdash_?}
\newcommand\polyj{\vdash_{\mathrm{poly}}}
\newcommand\intj{\vdash_{\mathrm{SGD}}}
\newcommand\gsj{\vdash_{\mathrm{unif}}}
\newcommand\dep{\Delta}

\newcommand{\mathhighlight}[1]{{\color{red!60!black} #1}}

\newcommand\ann{\alpha}
\newcommand\sann{\beta}
\newcommand\true{\mathbf t}
\newcommand\false{\mathbf f}

\DeclareMathOperator\valuefn{value}
\DeclareMathOperator\weightfn{weight}

\newcommand\infix[1]{\mathbin{#1}}

\newcommand\safet{\sigma}

\newcommand\smooth{\sigma_\eta}

\newcommand\transt[2]{\mathbf{transform}\,\sample_{#1}\,\mathbf{ by }\,#2}

\newcommand\nbt[4]{#1\langle#2,#3,#4\rangle}

\newcommand\trt{\Sigma}
\newcommand\addtr{\bullet}

\newcommand\PCFRp{\PCFReal_{>0}}

\newcommand\pdf{\mathrm{pdf}}

\newcommand\dist{\mathcal D}
\newcommand\mdist{\mathbf{\dist}}
\newcommand\nd{\mathcal N}
\newcommand\ed{\mathcal E}

\newcommand\app{\ensuremath{\mathbin{+\mkern-8mu+}}}

\newcommand\boundary\partial
\newcommand\interior\circ

\newcommand\unif{\xrightarrow{\textrm{unif.}}}
\newcommand\sau{\xrightarrow{\textrm{u.a.u.}}}

\newcommand\Jac{\mathbf{J}}

\newcommand\ball[1]{\mathbf B_{#1}(\mathbf 0)}

\newcommand{\tif}{\text{if }}

\newcommand{\tow}{\text{otherwise}}

\newcommand{\E}{\mathbb{E}}
\DeclareMathOperator{\dom}{dom}

\DeclareMathOperator*{\Leb}{Leb}

\newcommand*\diff{\mathop{}\!\mathrm{d}}

\newcommand\para{{\boldsymbol\theta}}
\newcommand\parasp{{\boldsymbol\Theta}}

\newcommand\lat\tr

\newcommand\obs{{\mathbf x}}
\newcommand\repx{\mathbf z}
\newcommand\rep{{\boldsymbol\phi}}

\newcommand\comp{\mathbin{@}}

\newcommand\flcomp{\mathbin{\ast}}
\newcommand\lrcomp{\mathbin{\odot}}

\newcommand\tr{\mathbf s}
\newcommand\from{:}
\newcommand\pto\rightharpoonup

\newcommand\etr{{[]}}

\newcommand\rel{\mathcal R}
\newcommand\ppred{\mathcal P}
\newcommand\qpred{\mathcal Q}

\newcommand\sem[1]{\llbracket #1\rrbracket}
\newcommand\sema[1]{\llbracket #1\rrbracket_{\eta}}

\newcommand\lyy{\hbox{\textsc{Lyy18} }}
\renewcommand\score{\hbox{\textsc{Score} }}
\newcommand\nested{\hbox{\textsc{Smooth} }}
\newcommand\reparam{\hbox{\textsc{Reparam }}}

\definecolor{oxblue}{RGB}{0,33,71}

%% file: sections/intro.tex

Probabilistic programming is a programming paradigm which has the vision to make statistical methods, in particular Bayesian inference, accessible to a wide audience. This is achieved by a separation of concerns: the domain experts wishing to gain statistical insights focus on modelling, whilst the inference is performed automatically. (In some recent systems \cite{pyro,DBLP:conf/pldi/Cusumano-Towner19} users can improve efficiency by writing their own inference code.)

In essence, probabilistic programming languages extend more traditional programming languages with constructs such as $\mathbf{score}$ or $\mathbf{observe}$ (as well as $\sample$) to define the prior $p(\repx)$ and likelihood $p(\obs\mid\repx)$. The task of inference is to derive the posterior $p(\repx\mid\obs)$, which is in principle governed by Bayes' law yet usually intractable.

Whilst the paradigm was originally conceived in the context of statistics and Bayesian machine learning,
probabilistic programming has in recent years proven to be a very fruitful subject for the programming language community.
Researchers have made significant theoretical contributions such as underpinning languages with rigorous (categorical) semantics \cite{SYWHK16,S17,HKSY17,VKS19,ETP14,DK20} and investigating the correctness of inference algorithms \cite{HNRS15,BLGS16,LYRY19}. The latter were mostly designed in the context of ``traditional'' statistics and features such as conditionals, which are ubiquitous in programming, pose a major challenge for correctness.

Inference algorithms broadly fall into two categories:
Markov chain Monte Carlo (MCMC), which yields a sequence of samples asymptotically approaching the true posterior, and variational inference.


\subsubsection*{Variational Inference.}
In the variational inference approach to Bayesian statistics \cite{Zhang2019,Murphy2012,B07,BKM17}, the problem of approximating difficult-to-compute posterior probability distributions is transformed to an optimisation problem.
The idea is to approximate the posterior probability $p(\repx \mid \obs)$ using a family of ``simpler'' densities $q_\para(\repx)$ over the latent variables $\repx$, parameterised by $\para$.
The optimisation problem is then to find the parameter $\para^\ast$ such that $q_{\para^\ast}(\repx)$ is ``closest'' to the true posterior $p(\repx \mid \obs)$.
Since the variational family may not contain the true posterior, $q_{\para^\ast}$ is an approximation in general. In practice, variational inference has proven to yield good approximations much faster than MCMC.

Formally, the idea is captured by minimising the \emph{KL-divergence} \cite{Murphy2012,B07} between the variational approximation and the true posterior.
This is equivalent to maximising the ELBO function, which only depends on the joint distribution $p(\obs,\repx)$ and \emph{not} the posterior, which we seek to infer after all:
\begin{equation}
  \label{eq:ELBO}
\mathrm{ELBO}_\para\defeq
\E_{\repx\sim q_\para{(\repx)}}[\log p(\obs,\repx)-\log q_\para(\repx)]
\end{equation}

\subsubsection*{Gradient Based Optimisation.}
In practice, variants of \emph{Stochastic Gradient Descent (SGD)} are frequently employed to solve optimisation problems of the following form: $\text{argmin}_\para\,\E_{\lat\sim q{(\lat)}}[f(\para,\lat)]$.
 In its simplest version, SGD follows Monte Carlo estimates of the gradient in each step:
\begin{align*}
  \para_{k+1}&\defeq\para_k-\gamma_k\cdot\underbrace{\frac 1 N\sum_{i=1}^N\nabla_\theta f\left(\para_k,\lat_k^{(i)}\right)}_{\text{gradient estimator}}
\end{align*}
where $\lat_k^{(i)}\sim q\left(\lat_k^{(i)}\right)$ and $\gamma_k$ is the \emph{step size}.

For the correctness of SGD it is crucial that the estimation of the gradient is \emph{unbiased}, i.e.\ correct in expectation:
\begin{align*}
  \E_{\lat^{(1)},\ldots,\lat^{(N)}\sim q}\left[\frac 1 N\sum_{i=1}^N\nabla_\theta f\left(\para,\lat^{(i)}\right)\right]=\nabla_\theta\E_{\lat\sim q{(\lat)}}[f(\para,\lat)]
\end{align*}
This property, which is about commuting differentiation and integration, can be established by the dominated convergence theorem \cite[Theorem 6.28]{K13}.

Note that we cannot directly estimate the gradient of the ELBO in \cref{eq:ELBO} \changed[dw]{with Monte Carlo} because the distribution w.r.t.\ which the expectation is taken also depends on the parameters. However, the so-called \emph{log-derivative trick} can be used to derive an unbiased estimate, which is known as the \emph{Score} or \emph{REINFORCE} estimator \cite{RGB14,WW13,DBLP:conf/icml/MnihG14}.

\subsubsection*{Reparameterisation Gradient.}

Whilst the score estimator has the virtue of being very widely applicable, it unfortunately suffers from high variance, which can cause SGD to yield very poor results\footnote{see e.g.\ \cref{fig:temperature-graph}}.

The \emph{reparameterisation gradient estimator}---the dominant approach in variational inference---reparameterises the latent variable $\repx$ in terms of a base random variable $\lat$ (viewed as the entropy source) via a diffeomorphic transformation $\rep_\para$, such as a location-scale transformation or cumulative distribution function.
For example, if the distribution of the latent variable $z$ is a Gaussian $\mathcal{N}(z \mid \mu, \sigma^2)$ with parameters $\para = \{\mu,\sigma\}$ then the location-scale transformation using the standard normal as the base distribution gives rise to the reparameterisation
\begin{align}
  \label{eq:locscale}
  z \sim \nd(z \mid \mu, \sigma^2)
  \iff
  z = \phi_{\mu,\sigma}(s), \quad s \sim \nd(0,1).
\end{align}
where $\phi_{\mu,\sigma}(s)\defeq s\cdot\sigma+\mu$.
The key advantage of this setup (often called ``reparameterisation trick'' \cite{DBLP:journals/corr/KingmaW13,DBLP:conf/icml/TitsiasL14,DBLP:conf/icml/RezendeMW14}) is that we have removed the dependency on $\para$ from the distribution w.r.t.\ which the expectation is taken. Therefore, we can now differentiate (by backpropagation) with respect to the parameters $\para$ of the variational distributions using Monte Carlo simulation with draws from the base distribution $\lat$.
Thus, succinctly,
we have
\begin{align*}
\nabla_\para \,\E_{\repx\sim q_\para(\repx)}[f(\para,\repx)]
=
\nabla_\para \,\E_{\lat\sim q(\lat)}[f(\para,\rep_\para(\lat))]
&=
\E_{\lat\sim q(\lat)}[\nabla_\para \, f(\para,\rep_\para(\lat))]
\end{align*}

The main benefit of the reparameterisation gradient estimator is that it has a significantly lower variance than the score estimator, resulting in faster convergence.

\subsubsection{Bias of the Reparameterisation Gradient.}
Unfortunately, the reparameterisation gradient estimator
is biased for non-differentiable models \cite{LYY18}, which are readily expressible in programming languages with conditionals:
\begin{example}
  \label{ex:biased}
  The counterexample in \cite[Proposition~2]{LYY18}, where the objective function is the ELBO for a non-differentiable model, can be simplified to
  \begin{align*}
    f(\theta,s)&= -0.5\cdot\theta^2+
    \begin{cases}
      0&\tif s+\theta<0\\
      1&\tow
    \end{cases}
  \end{align*}
  Observe that (see \cref{fig:biased}):
  \[
    \nabla_\theta\,\E_{s\sim\nd(0,1)}\left[f(\theta,s)\right]
    =
    -\theta+\nd(-\theta\mid 0,1)
    \neq
    -\theta
    =
    \E_{s\sim\nd(0,1)}\left[\nabla_\theta f(\theta,s)\right]
  \]
\end{example}
Crucially \emph{this may compromise convergence to critical points or maximisers}:
even if we can find a point where the gradient estimator vanishes, it may not be a critical point (let alone optimum) of the original optimisation problem (cf.~\cref{fig:exelbo})


\begin{figure}[t]
  \centering
  \begin{subfigure}[h]{.43\textwidth}
    \centering
    \begin{tikzpicture}
      \begin{axis}[
        height=5cm,
        axis lines = middle,
        xlabel = \(\theta\),
        ]
        \addplot [
        domain=-1:1,
        samples=100,
        color=oxred,
        dashed,
        thick
        ]
        {-x};
        \addplot [
        update limits=false,
        domain=-1:1,
        samples=100,
        color=oxgreen,
        thick
        ]
        {-x+1/sqrt(2*pi)*exp(-0.5*x*x)};
      \end{axis}
    \end{tikzpicture}
    \caption{Dashed red: biased estimator $\E_{s\sim\nd(0,1)}\left[\nabla_\theta f(\theta,s)\right]$, solid green: true gradient $\nabla_\theta\,\E_{s\sim\nd(0,1)}\left[f(\theta,s)\right]$.}
    \label{fig:biased}
  \end{subfigure}
  \qquad
  \begin{subfigure}[h]{.5\textwidth}
    \includegraphics[width=\linewidth]{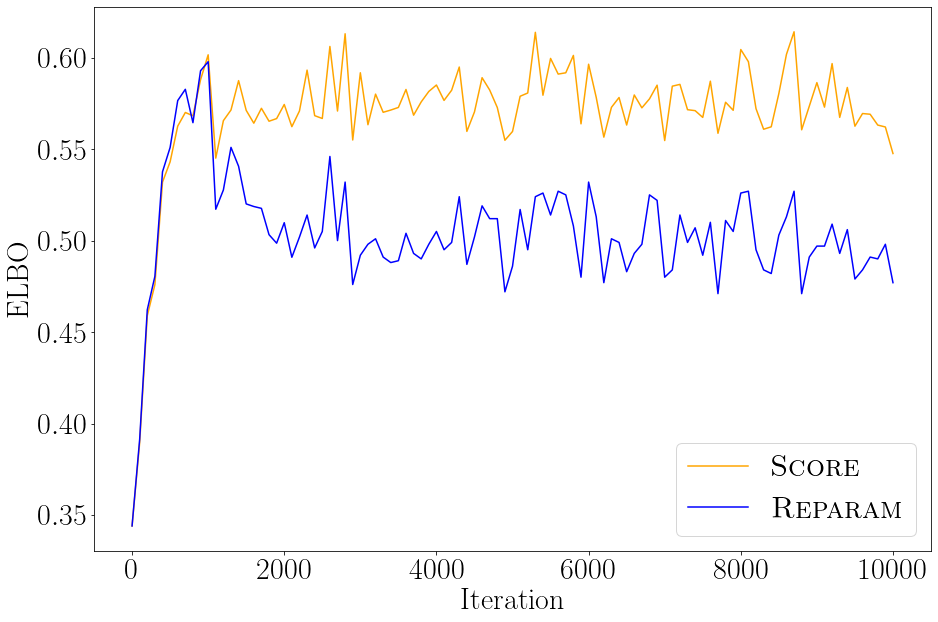}
      \caption{ELBO trajectories (higher means better) obtained with our implementation (cf.~\cref{sec:empirical})}
      \label{fig:exelbo}
  \end{subfigure}
  \caption{Bias of the reparameterisation gradient estimator for \cref{ex:biased}.}
\end{figure}
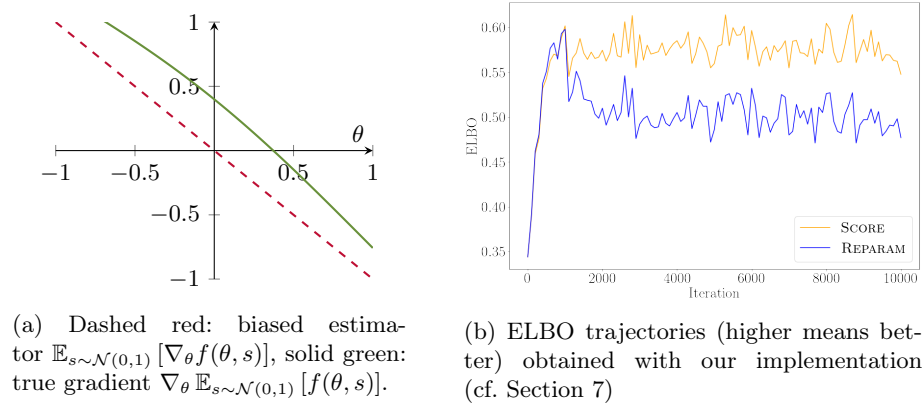

\subsection*{Informal Approach}
As our starting point we take a variant of the simply typed lambda calculus with reals, conditionals and a sampling construct.
We abstract the optimisation of the ELBO to the following generic optimisation problem
\begin{equation}
  \text{argmin}_\para\,\E_{\lat\sim\dist}[\sem M(\para,\lat)]
\end{equation}
where $\sem M$ is the value function \cite{BLGS16,MOPW}
 of \changed[dw]{a program} $M$ and $\dist$ is independent of the parameters $\para$ and it is determined by the distributions from which $M$ samples.
Owing to the presence of conditionals, the function $\sem M$ may not be continuous, let alone differentiable.

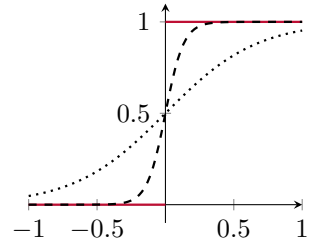
\begin{wrapfigure}{r}{0.3\linewidth}
  \vspace*{-1.1cm}
  \begin{tikzpicture}
    \begin{axis}[
      ymin=-0.1,
      ymax=1.1,
      height=4.5cm,
      axis lines = middle,
      ]
      \addplot [
      domain=-1:0,
      samples=100,
      color=oxred,
      line width=0.3mm
      ]
      {{0}};
      \addplot [
      domain=0:1,
      samples=100,
      color=oxred,
      line width=0.3mm
      ]
      {{1}};

      \addplot [
      domain=-1:1,
      samples=100,
      line width=0.3mm,
      dotted
      ]
      {{1/(1+exp(-3*x))}};

      \addplot [
      update limits=false,
      domain=-1:1,
      samples=100,
      line width=0.3mm,
      dashed
      ]
      {1/(1+exp(-15*x))};
    \end{axis}
  \end{tikzpicture}
  \caption{(Logistic) sigmoid function $\smooth$ (dotted: $\eta=\frac 1 3$, dashed: $\eta=\frac 1{15}$) and the Heaviside step function (red, solid).}
  \label{fig:sigf}
\end{wrapfigure}
\cref{ex:biased} can be expressed as
\begin{align*}
  (\lambda z\ldotp-0.5\cdot\theta^2+(\ifc z 0 1))\,(\sample_\nd+\theta)
\end{align*}

Our approach is based on a denotational semantics $\sema{(-)}$ (for accuracy coefficient $\eta > 0$) of programs in the (new) cartesian closed category $\VFr$, which generalises smooth manifolds and extends Fr\"olicher spaces (see e.g.~\cite{FK88,S11}) with a vector space structure.

Intuitively, we replace the Heaviside step-function usually arising in the interpretation of conditionals by smooth approximations.
In particular, we interpret the conditional of \cref{ex:biased} as
\begin{gather*}
  \sema {\ifc{s+\theta}{\underline 0}{\underline 1}}(\theta,s)\defeq \smooth(s+\theta)
\end{gather*}
where $\smooth$ is a smooth function.
For instance we can choose $\smooth(x)\defeq\sigma(\frac x\eta)$ where $\sigma(x) \defeq \frac 1{1+\exp(-x)}$ is the (logistic) sigmoid function (cf. \cref{fig:sigf}).
Thus, the program $M$ is interpreted by a smooth function $\sema M$, for which the reparameterisation gradient may be estimated unbiasedly.
  Therefore, we apply stochastic gradient descent on the smoothed program.


\subsection*{Contributions}
The high-level contribution of this paper is laying a theoretical foundation for correct yet efficient (variational) inference for probabilistic programming. We employ a smoothed interpretation of programs to obtain unbiased (reparameterisation) gradient estimators and establish technical pre-conditions by type systems.
In more detail:
\begin{enumerate}
  \item We present a simple (higher-order) programming language with conditionals.
  We employ \emph{trace types} to capture precisely the samples drawn in a fully eager call-by-value evaluation strategy.
  \item We endow our language with both a (measurable) denotational value semantics and a smoothed (hence approximate) value semantics. For the latter we furnish a categorical model based on Fr\"olicher spaces.
  \item We develop type systems enforcing vital technical pre-conditions: unbiasedness of the reparameterisation gradient estimator and the correctness of stochastic gradient descent, as well as the uniform convergence of the smoothing to the original problem.
  Thus, our smoothing approach in principle yields correct solutions up to arbitrary error tolerances.
  \item We conduct an empirical evaluation demonstrating that our approach exhibits a similar convergence to an unbiased correction of the reparameterised gradient estimator by \cite{LYY18} -- our main baseline.
  However our estimator is simpler and more efficient: it is faster and attains orders of magnitude reduction in work-normalised variance.
\end{enumerate}

\paragraph{Outline.}
In the next section we introduce a simple higher-order probabilistic programming language, its denotational value semantics and operational semantics; Optimisation \cref{p:opt} is then stated.
\cref{sec:smoothing} is devoted to a smoothed denotational value semantics, and we state the Smooth Optimisation \cref{p:optsm}.
In \cref{sec:sgd,sec:conv} we develop annotation based type systems enforcing the correctness of SGD and the convergence of the smoothing, respectively.
Related work is briefly discussed in \cref{sec:related} before
 we present the results of our empirical evaluation in \cref{sec:empirical}.
We conclude in \cref{sec:conclusion} and discuss future directions.

\paragraph{Notation.}
We use the following conventions: bold font for vectors and lists, $\app$ for concatenation of lists, $\nabla_\para$ for gradients (w.r.t.\ $\para$),$[\phi]$ for the Iverson bracket of a predicate $\phi$ and calligraphic font for distributions, in particular $\nd$ for normal distributions.
Besides, we highlight noteworthy items using $\mathhighlight{\mathrm{red}}$.

%% file: sections/PL.tex

In this section, we introduce our programming language, which is the simply-typed lambda calculus with reals, augmented with conditionals and sampling from continuous distributions.

\subsection{Syntax}
The \emph{raw terms} of the programming language are defined by the grammar:
\begin{align*}
  M::= x&\mid\theta_i\mid\underline r\mid \underline +\mid\underline\cdot\mid\PCFmin\mid\PCFinv\mid\underline\exp\mid\underline\log\\
  &\mid \ifc MMM\mid\sample_\dist\mid\lambda x\ldotp M\mid M\,M
\end{align*}
where $x$ and $\theta_i$ respectively range over (denumerable collections of) \emph{variables} and \emph{parameters}, $r \in \Real$, and $\dist$ is a probability distribution over $\Real$ (potentially with a support which is a strict subset of $\Real$).
As is customary we use infix, postfix and prefix notation: $M\infix\PCFplus N$ (addition), $M\infix\PCFmul N$ (multiplication), $M\PCFinv$ (inverse), and $\PCFmin M$ (numeric negation). We frequently omit the underline to reduce clutter.

\begin{example}[Encoding the ELBO for Variational Inference]
  \label{ex:varinf}
  We consider the example used by \cite{LYY18} in their Prop.~2 to prove the biasedness of the reparameterisation gradient. (In \cref{ex:biased} we discussed a simplified version thereof.)
  The joint density is
  \begin{align*}
    p(z)&\defeq \nd(z\mid 0,1)\cdot
    \begin{cases}
      \nd(0\mid -2,1)&\tif z<0\\
      \nd(0\mid 5,1)&\tow
    \end{cases}
  \end{align*}
  and they use a variational family with density $q_\theta(z)\defeq\nd(z\mid \theta,1)$, which is reparameterised using a standard normal noise distribution and transformation $s\mapsto s+\theta$.

  First, we define an auxiliary term for the pdf of normals with mean $m$ and standard derivation $s$:
  \begin{align*}
    N\equiv \lambda x,m,s\ldotp \big(\underline{\sqrt{2\pi}}\cdot s\big)^{-1}\cdot\PCFexp\,\Big(\underline{-0.5}\cdot\big((x+(-m))\cdot s^{-1}\big)^2\Big)
  \end{align*}
  Then, we can define
  \begin{align*}
    M&\equiv \big(\lambda z\ldotp \underbrace{\PCFlog\,(N\,z\,\underline 0\,\underline 1)+(\ifc {z}{\PCFlog\,(N\,\underline 0\,\underline{(-2)}\,{\underline 1})}{\PCFlog\,(N\,\underline 0\,\underline 5\,\underline 1)})}_{\log p}-\\
    &\hspace*{5cm}\underbrace{\PCFlog\,(N\,z\,\theta\,\underline 1)}_{\log q_\theta}\big)
    \,\big(\sample_\nd+\theta\big)
  \end{align*}
\end{example}

\subsection{A Basic Trace-Based Type System}
\label{sec:basic type system}

\emph{Types} are generated from \emph{base types} ($\PCFReal$ and $\PCFRp$, the reals and positive reals) and \emph{trace types} (typically $\Sigma$, which is a finite list of probability distributions) as well as by a \emph{trace-based function space} constructor of the form $\tau \addtr \Sigma \to \tau'$.
Formally types are defined by the following grammar:
\begin{align*}
  &\textbf{trace types}&\trt&::= [\dist_1,\ldots,\dist_n] \qquad n \geq 0\\
  &\textbf{base types}&\iota&::=\PCFReal\mid\PCFRp\\
  &\textbf{safe types}&\safet&::=\iota\mid\safet\addtr\mathhighlight{[]}\to\safet\\
  &\textbf{types}&\tau&::=\iota\mid\tau\addtr\trt\to\tau
\end{align*}
where $\dist_i$ are probability distributions.
Intuitively a trace type is a description of the space of execution traces of a probabilistic program.
Using trace types, a distinctive feature of our type system is that a program's type precisely characterises the space of its possible execution traces \cite{DBLP:journals/pacmpl/LewCSCM20}.
We use list concatenation notation $\app$ for trace types, and the shorthand $\tau_1\to\tau_2$ for function types of the form $\tau_1\addtr[]\to\tau_2$.
Intuitively, a term has type $\tau\addtr\trt\to\tau'$ if, when given a value of type $\tau$, it reduces to a value of type $\tau'$ using all the samples in $\trt$.

Dual context \emph{typing judgements} of the form, $\Gamma\mid\trt\vdash M:\tau$, are defined in \cref{fig:stype}, where $\Gamma = x_1:\tau_1, \cdots, x_n:\tau_n, \theta_1 : \tau_1', \cdots, \theta_m : \tau_m'$ is a finite map describing a set of variable-type and parameter-type bindings; and the trace type $\trt$ precisely captures the distributions from which samples are drawn in a (fully eager) call-by-value evaluation of the term $M$.


The subtyping of types, as defined in \cref{fig:subtyping}, is essentially standard; for contexts, we define $\Gamma\sqsubseteq\Gamma'$ if for every $x:\tau$ in $\Gamma$ there exists $x:\tau'$ in $\Gamma'$ such that $\tau'\sqsubseteq\tau$.

Trace types are unique (cf.~\cref{app:basict}):
\begin{restatable}{lemma}{unitt}
  \label{lem:unitt}
  If $\Gamma\mid\trt\vdash M:\tau$ and $\Gamma\mid\trt'\vdash M:\tau'$ then $\trt=\trt'$.
\end{restatable}

A term has \emph{safe type} $\safet$ if it does not contain $\sample_\dist$ or $\safet$ is a base type.
Thus, perhaps slightly confusingly, we have ${} \mid[\dist]\vdash\sample_\dist:\PCFReal$, and $\PCFReal$ is considered a safe type.
Note that we use the metavariable $\safet$ to denote safe types.

\paragraph{Conditionals.}
The branches of conditionals must have a safe type.
Otherwise it would not be clear how to type terms such as
\begin{align*}
  M&\equiv \ifc x {(\lambda x\ldotp \sample_\nd)}{(\lambda x\ldotp \sample_\ed+\sample_\ed)}\\
  N&\equiv (\lambda f\ldotp f\,(f\,\sample_\nd))\,M
\end{align*}
because the branches draw a different number of samples from different distributions, and have types $\PCFReal\addtr[\nd]\to\PCFReal$ and $\PCFReal\addtr[\ed,\ed]\to\PCFReal$, respectively.
However, for $M'\equiv \ifc x{\sample_\nd}{\sample_\ed+\sample_\ed}$ we can (safely) type
\begin{align*}
   x:\PCFReal\mid[\nd,\ed,\ed]&\vdash M':\PCFReal\\
   \mid[]&\vdash \lambda x\ldotp M':\PCFReal\addtr[\nd,\ed,\ed]\to\PCFReal\\
   \mid[\nd,\nd,\ed,\ed,\nd,\ed,\ed]&\vdash (\lambda f\ldotp f\,(f\,\sample_\nd))\,(\lambda x\ldotp M'):\PCFReal
\end{align*}


\begin{figure}[t]
\begin{framed}
  \begin{subfigure}{\linewidth}
  \begin{gather*}
    \infer{\iota\sqsubseteq\iota}{}\qquad\infer{\PCFRp\sqsubseteq\PCFReal}{}\qquad
      \infer{(\tau_1\addtr\trt\to\tau_2)\sqsubseteq(\tau_1'\addtr\trt\to\tau_2')}{\tau'_1\sqsubseteq\tau_1&
      \tau_2\sqsubseteq\tau'_2}
  \end{gather*}
  \caption{Subtyping}
  \label{fig:subtyping}
  \end{subfigure}

 \begin{subfigure}{\linewidth}
  \begin{gather*}
    \infer[\Gamma\sqsubseteq\Gamma',
    \tau\sqsubseteq\tau']{\Gamma'\mid\trt\vdash M:\tau'}{\Gamma\mid\trt\vdash M:\tau}\qquad
    \infer{x:\tau\mid[]\vdash x:\tau}{}\\
    \infer[r\in\Real]{\mid[]\vdash \underline r:\PCFReal}{}\qquad
    \infer[r\in\Real_{>0}]{\mid[]\vdash \underline r:\PCFRp}{}\\
    \infer[\circ\in\{+,\cdot\}]{\mid[]\vdash \PCFcirc:\PCFReal\to\PCFReal\to\PCFReal}{}\qquad
    \infer[\circ\in\{+,\cdot\}]{\mid[]\vdash \PCFcirc:\PCFRp\to\PCFRp\to\PCFRp}{}\\
    \infer{\mid[]\vdash \underline-:\PCFReal\to\PCFReal}{}\qquad
    \infer{\mid[]\vdash\PCFinv:\PCFReal_{\mathhighlight{>0}}\to\PCFRp}{}\\
    \infer{\mid[]\vdash \PCFexp:\PCFReal\to\PCFRp}{}\qquad
    \infer{\mid[]\vdash \PCFlog:\PCFReal_{\mathhighlight{>0}}\to\PCFReal}{}\\
    \infer{\Gamma\mid\trt\app\trt'\app\trt''\vdash\ifc LMN:\mathhighlight\safet}{\Gamma\mid\trt\vdash L:\PCFReal&\Gamma\mid\trt'\vdash M:\mathhighlight\safet&\Gamma\mid\trt''\vdash N:\mathhighlight\safet}\qquad
    \infer{\mid[\dist]\vdash\sample_{\mathcal\dist}:\PCFReal}{}\\
    \infer{\Gamma\mid[]\vdash\lambda y\ldotp M:\tau_1\addtr\trt\to\tau_2}{\Gamma,y:\tau_1\mid\trt\vdash M:\tau_2}\qquad
    \infer{\Gamma\mid\trt_1\app\trt_2\app\trt_3\vdash M\,N:\tau_2}{\Gamma\mid\trt_1\vdash M:\tau_1\addtr\trt_3\to\tau_2&\Gamma\mid\trt_2\vdash N:\tau_1}
  \end{gather*}
  \caption{Typing judgments}
  \label{fig:stype}
\end{subfigure}
  \caption{A Basic Trace-based Type System}
\end{framed}
\end{figure}

\begin{example}
  Consider the following terms:
\begin{align*}
  L&\equiv\lambda x\ldotp\sample_\nd+\sample_\nd\\
  M&\equiv\ifc x{(\lambda y\ldotp y+y)\,\sample_\nd}{(\sample_\nd+\sample_\nd)}
\end{align*}
   We can derive the following typing judgements:
  \begin{align*}
    \mid [] &\vdash L:\PCFRp\addtr [\nd,\nd]\to\PCFReal\\
    x : \PCFRp \mid [\nd,\nd,\nd]&\vdash M:\PCFReal\\
    \mid[]&\vdash\lambda x\ldotp M:\PCFRp\addtr[\nd,\nd,\nd]\to\PCFReal\\
    \mid[\nd,\nd,\nd,\nd]&\vdash(\lambda x\ldotp M)\,\sample_\nd:\PCFReal\\
    {} \mid [\nd, \nd] &\vdash (\lambda f\ldotp f \, (f \, 0))\,(\lambda x\ldotp \sample_\nd) : \PCFReal
  \end{align*}
  Note that $\ifc x{(\lambda x\ldotp\sample_\nd)}{(\lambda x\ldotp x)}$ is \emph{not} typable. 
\end{example}

\subsection{Denotational Value Semantics}
\label{sec:densem}

Next, we endow our language with a (measurable) value semantics.
It is well-known that the category of measurable spaces and measurable functions is not cartesian-closed \cite{Aumann61}, which means that there is no interpretation of the lambda calculus as measurable functions.
These difficulties led \cite{Heunen2017c} to develop the category $\mathbf{QBS}$ of \emph{quasi-Borel spaces}. In \cref{app:densem} we recall the definition. Notably, morphisms can be combined piecewisely, which we need for conditionals.

We interpret our programming language in the category $\qbs$ of quasi-Borel spaces.
Types are interpreted as follows: 
\[
  \begin{array}{c}
  \sem{\PCFReal} \defeq (\Real, M_\Real) \qquad
  \sem{\PCFRp} \defeq (\Realp, M_{\Realp}) \qquad
  \sem{[\dist_1,\ldots,\dist_n]} \defeq (\Real, M_\Real)^n\\
  \sem{\tau_1\addtr\trt\to\tau_2} \defeq\sem{\tau_1}\times\sem\trt\Rightarrow\sem{\tau_2}
\end{array}
\]
where $M_\Real$ is the set of measurable functions $\Real \to \Real$; similarly for $M_{\Realp}$.
(As for trace types, we use list notation (and list concatenation) for traces.)

We first define a handy helper function for interpreting application.
For $f:\sem\Gamma\times\Real^{n_1}\Rightarrow\sem{\tau_1\addtr\trt_3\to\tau_2}$ and $g:\sem\Gamma\times\Real^{n_2}\Rightarrow\sem{\tau_1}$ define
\begin{align*}
  f\comp g:\sem\Gamma\times\Real^{n_1+n_2+|\trt_3|}&\Rightarrow\sem{\tau_2}\\
  (\gamma,\tr_1\app\tr_2\app\tr_3)&\mapsto f(\gamma,\tr_1)(g(\gamma,\tr_2),\tr_3)&\tr_1\in\Real^{n_1},\tr_2\in\Real^{n_2},\tr_3\in\Real^{|\trt_3|}
\end{align*}

We interpret terms-in-context, $\sem{\Gamma\mid\trt\vdash M:\tau}:\sem\Gamma\times\sem\trt\to\sem\tau$, as follows: \dw{give all cases} \dw{subtyping}
\begin{align*}
  \sem{\Gamma\mid[\dist]\vdash\sample_\dist:\PCFReal}(\gamma,[s])&\defeq s\\
  \sem{\Gamma\mid[]\vdash\lambda y\ldotp M:\tau_1\addtr\trt\to\tau_2}(\gamma,\etr)&\defeq\\
  &\hspace*{-2cm}
  (v,\tr)\in\sem{\tau_1}\times\sem\trt\mapsto\sem{\Gamma,x:\tau_1\mid\trt\vdash M:\tau_2}((\gamma,v),\tr)\\
  \sem{\Gamma\mid\trt_1\app\trt_2\app\trt_3\vdash M\,N:\tau}&\defeq\\
  &\hspace*{-2cm}
  \sem{\Gamma\mid\trt_1\vdash M:\tau_1\addtr\trt_3\to\tau_2}\comp\sem{\Gamma\mid\trt_2\vdash N:\tau_1}\\
  \sem{\Gamma\mid\trt_1\app\trt_2\app\trt_3\vdash\ifc LMN:\tau}(\gamma,\tr_1\app\tr_2\app\tr_3))\defeq\hspace{-6cm}&\\
  &\hspace*{-2cm}
  \begin{cases}
  \sem{\Gamma\mid\trt_2\vdash M:\tau}(\gamma,\tr_2)&\tif \sem{\Gamma\mid\trt_1\vdash L:\PCFReal}(\gamma,\tr_1)<0\\
  \sem{\Gamma\mid\trt_3\vdash N:\tau}(\gamma,\tr_3)&\tow
\end{cases}
\end{align*}

It is not difficult to see that this interpretation of terms-in-context is well-defined and total.
For the conditional clause, we may assume that the trace type and the trace are presented as partitions $\trt_1 \app  \trt_2 \app  \trt_3$ and $\tr_1 \app  \tr_2 \app  \tr_3$ respectively.
This is justified because it follows from the judgement $\Gamma\mid\trt_1\app\trt_2\app\trt_3\vdash\ifc LMN:\tau$ that $\Gamma \mid \Sigma_1 \vdash L : R$, $\Gamma \mid \Sigma_2 \vdash M : \sigma$ and $\Gamma \mid \Sigma_3 \vdash N : \sigma$ are provable; and we know that each of $\Sigma_1, \Sigma_2$ and $\Sigma_3$ is unique, thanks to \cref{lem:unitt}; their respective lengths then determine the partition $\tr_1 \app  \tr_2 \app  \tr_3$.
Similarly for the application clause, the components $\Sigma_1$ and $\Sigma_2$ are determined by \cref{lem:unitt}, and $\Sigma_3$ by the type of $M$.

%

\subsection{Relation to Operational Semantics}
\label{sec:opsem}
We can also endow our language with a big-step CBV sampling-based semantics similar to \cite{BLGS16,MOPW}, as defined in \cref{fig:ops} of \cref{app:PL}.
We write $M\Downarrow_w^\tr V$ to mean that $M$ reduces to value $V$, which is a real constant or an abstraction, using the execution trace $\tr$ and accumulating weight $w$.

Based on this, we can define the \defn{value-} and \defn{weight}-functions:
\begin{align*}
  \valuefn_M(\tr)&\defeq \begin{cases}
    V&\tif M\Downarrow^\tr_w V\\
    \mathrm{undef}&\tow
\end{cases}
&
\weightfn_M(\tr)&\defeq
\begin{cases}
  w&\tif M\Downarrow^\tr_w V\\
  0&\tow
\end{cases}
\end{align*}

Our semantics is a bit non-standard in that for conditionals we evaluate both branches eagerly. The technical advantage is that for every (closed) term-in-context, ${} \mid [\dist_1, \cdots, \dist_n] \vdash M : \iota$, $M$ reduces to a (unique) value using exactly the traces of the length encoded in the typing, i.e.,~$n$.



So in this sense, the operational semantics is ``total'': there is no divergence.
Notice that there is no partiality caused by partial primitives such as $1/x$, thanks to the typing.

Moreover there is a simple connection to our denotational value semantics:

%
%

\begin{proposition}
  Let ${} \mid[\dist_1,\ldots,\dist_n]\vdash M:\iota$. Then
  \begin{enumerate}
    \item $\dom(\valuefn_M)=\Real^n$
    \item $\underline{\sem M}=\valuefn_M$
    \item $\weightfn_M(\tr)=\prod_{j=1}^n\pdf_{\dist_j}(s_j)$
  \end{enumerate}
\end{proposition}

\subsection{Problem Statement}
We are finally ready to formally state our optimisation problem:
%


\begin{problem}
  \label{p:opt}
  $ $ Optimisation
  \normalfont
  \vspace*{-2pt}
  \par\addvspace{.5\baselineskip}
    \noindent
    \begin{tabularx}{\textwidth}{@{\hspace{\parindent}} l X c}
      \textbf{Given:} & term-in-context, $\theta_1 : \iota_1, \cdots, \theta_m : \iota_m \mid [\dist_1,\ldots,\dist_n]\vdash M:\PCFReal$
      \\[2pt]
      \textbf{Find:} & $\text{argmin}_\para \ \E_{s_1\sim\dist_1,\ldots,s_n\sim\dist_n}\left[\sem M(\para,\lat)\right]$
    \end{tabularx}
    \par\addvspace{.5\baselineskip}
\end{problem}


%% file: sections/smooth.tex

Now we turn to our smoothed denotational value semantics,
which we use to avoid the bias in the reparameterisation gradient estimator.
It is parameterised by a family of smooth functions $\smooth:\Real\to[0,1]$.
Intuitively, we replace the Heaviside step-function arising in the interpretation of conditionals by smooth approximations (cf.~\cref{fig:sigf}). In particular, conditionals $\ifc z{\underline 0}{\underline 1}$ are interpreted as $z\mapsto\smooth(z)$ rather than $[z\geq 0]$ (using Iverson brackets).

Our primary example is $\smooth(x)\defeq\sigma(\frac x\eta)$, where $\sigma$ is the (logistic) sigmoid $\sigma(x)\defeq\frac 1{1+\exp(-x)}$, see \cref{fig:sigf}.
Whilst at this stage no further properties other than smoothness are required, we will later need to restrict $\smooth$ to have good properties, in particular to convergence to the Heaviside step function.

As a categorical model we propose \emph{vector Frölicher spaces} $\VFr$, which (to our knowledge) is a new construction, affording a simple and direct interpretation of the smoothed conditionals.

\subsection{Fr\"olicher Spaces}
We recall the definition of
Fr\"olicher spaces, which generalise smooth spaces\footnote{$C^\infty(\Real,\Real)$ is the set of smooth functions $\Real\to\Real$}:
  A \emph{Fr\"olicher space} is a triple $(X,\mathcal C_X,\mathcal F_X)$ where
     $X$ is a set,
    $\mathcal C_X\subseteq\Set(\Real,X)$ is a set of \emph{curves} and
    $\mathcal F_X\subseteq\Set(X,\Real)$ is a set of \emph{functionals}.
  satisfying
  \begin{enumerate}
    \item if $c\in\mathcal C_X$ and $f\in\mathcal F_X$ then $f\circ c\in C^\infty(\Real,\Real)$
    \item if $c:\Real\to X$ such that for all $f\in\mathcal F_X$, $f\circ c\in C^\infty(\Real,\Real)$ then $c\in\mathcal C_X$
    \item if $f:X\to\Real$ such that for all $c\in\mathcal C_X$, $f \circ c\in C^\infty(\Real,\Real)$ then $f\in\mathcal F_X$.
  \end{enumerate}
  A \emph{morphism} between Fr\"olicher spaces $(X,\mathcal C_X,\mathcal F_X)$ and $(Y,\mathcal C_Y,\mathcal F_Y)$ is a map $\phi\from X\to Y$ satisfying $f\circ\phi\circ c\in C^\infty(\Real,\Real)$ for all $f\in\mathcal F_Y$ and $c\in\mathcal C_X$.

Fr\"olicher spaces and their morphisms constitute a category $\Fr$, which is well-known to be cartesian closed \cite{FK88,S11}.

\subsection{Vector Frölicher Spaces}
To interpret our programming language smoothly we would like to interpret conditionals as $\smooth$-weighted convex combinations of its branches:
\begin{align}
  \sema{\ifc LMN}(\gamma,\tr_1\app\tr_2\app\tr_3)&\defeq\notag\\
  &\hspace*{-5cm}\smooth(-\sema{L}(\gamma,\tr_1))\cdot\sema{M}(\gamma,\tr_2)+\smooth(\sema{L}(\gamma,\tr_1))\cdot\sema{N}(\gamma,\tr_3)
  \label{eq:smcond}
\end{align}
By what we have discussed so far, this only makes sense if the branches have ground type because Fr\"olicher spaces are not equipped with a vector space structure but we take weighted combinations of morphisms. In particular if $\phi_1,\phi_2:X\to Y$ and $\alpha:X\to\Real$ are morphisms then $\alpha \, \phi_1+\phi_2$ ought to be a morphism too.
Therefore, we enrich Fr\"olicher spaces with an additional vector space structure:
\begin{definition}
  A $\Real$-\emph{vector Fr\"olicher space} is a Fr\"olicher space $(X,\mathcal C_X,\mathcal F_X)$ such that $X$ is an $\Real$-vector space and whenever $c,c'\in\mathcal C_X$ and $\alpha \in C^\infty(\Real,\Real)$ then $\alpha \, c+c'\in\mathcal C_X$ (defined pointwise).

  A \emph{morphism} between $\Real$-\emph{vector Fr\"olicher spaces} is a morphism between Fr\"olicher spaces, i.e.\ $\phi:(X,\mathcal C_X,\mathcal F_X)\to (Y,\mathcal C_Y,\mathcal F_Y)$ is a morphism if for all $c\in\mathcal C_X$ and $f\in\mathcal F_Y$, $f\circ\phi\circ c\in C^\infty(\Real,\Real)$.
\end{definition}
$\Real$-\emph{vector Fr\"olicher space} and their morphisms constitute a category $\VFr$. There is an evident forgetful functor fully faithfully embedding $\VFr$ in $\Fr$.
Note that the above restriction is a bit stronger than requiring that $\mathcal C_X$ is also a vector space. ($\alpha$ is not necessarily a constant.)
The main benefit is the following, which is crucial for the interpretation of conditionals as in \cref{eq:smcond}:
\begin{lemma}
   If $\phi_1,\phi_2\in\VFr(X,Y)$ and $\alpha\in\VFr(X,\Real)$ then $\alpha \, \phi_1+\phi_2\in\VFr(X,Y)$ (defined pointwisely).
\end{lemma}
\begin{proof}
  Suppose $c\in\mathcal C_X$ and $f\in\mathcal F_Y$. Then $(\alpha_1 \, \phi_1+\phi_2)\circ c=(\alpha\circ c)\cdot(\phi_1\circ c)+(\phi_2\circ c)\in\mathcal C_Y$ (defined pointwisely) and the claim follows.
\end{proof}

Similarly as for Fr\"olicher spaces, if $X$ is an $\Real$-vector space then any $\mathcal C\subseteq\Set(X,\Real)$ \emph{generates} a $\Real$-vector Fr\"olicher space $(X,\mathcal C_X,\mathcal F_X)$, where
\begin{align*}
  \mathcal F_X&\defeq\{f:X\to\Real\mid\forall c\in\mathcal C\ldotp f\circ c\in C^\infty(\Real,\Real)\}\\
  \widetilde{\mathcal C}_X&\defeq\{c:\Real\to X\mid\forall f\in\mathcal F_X\ldotp f\circ c\in C^\infty(\Real,\Real)\}\\
  \mathcal C_X&\defeq\left\{\sum_{i=1}^n\alpha_i \, c_i\mid n\in\nat,\forall i\leq n\ldotp \alpha_i\in C^\infty(\Real,\Real), c_i\in\widetilde{\mathcal C}_X\right\}
\end{align*}
Having modified the notion of Frölicher spaces generated by a set of curves, the proof for cartesian closure carries over (more details are provided in \cref{app:smooth}) and we conclude:
\begin{restatable}{proposition}{vfrccc}
  $\VFr$ is cartesian closed.
\end{restatable}

\subsection{Smoothed Interpretation}
We have now discussed all ingredients to interpret our language (smoothly) in the cartesian closed category $\VFr$.
We call $\sema M$ the $\eta$-\defn{smoothing} of $\sem M$ (or of $M$, by abuse of language).
 The interpretation is mostly standard and follows \cref{sec:densem}, except for the case for conditionals.
 The latter is given by \cref{eq:smcond}, for which the additional vector space structure is required. \dw{details in appendix}

Finally, we can phrase a smoothed version of our Optimisation \cref{p:opt}:
%


\begin{problem}
  \label{p:optsm}
  $ $ $\eta$-Smoothed Optimisation
  \normalfont
  \vspace*{-2pt}
  \par\addvspace{.5\baselineskip}
    \noindent
    \begin{tabularx}{\textwidth}{@{\hspace{\parindent}} l X c}
      \textbf{Given:} & term-in-context, $\theta_1 : \iota_1, \cdots, \theta_m : \iota_m \mid [\dist_1,\ldots,\dist_n]\vdash M:\PCFReal$, and \emph{accuracy} coefficient $\eta>0$
      \\[2pt]
      \textbf{Find:} & $\text{argmin}_\para \ \E_{s_1\sim\dist_1,\ldots,s_n\sim\dist_n}\left[\sema M(\para,\lat)\right]$
    \end{tabularx}
    \par\addvspace{.5\baselineskip}
\end{problem}


%% file: sections/SGD.tex

Next, we apply stochastic gradient descent (SGD) with the reparameterisation gradient estimator to the smoothed problem (for the batch size $N=1$):
\begin{align}
  \label{eq:sgd}
  \para_{k+1}&\defeq\para_k-\gamma_k\cdot\nabla_\theta\sema M\left(\para_k,\lat_k\right) &\lat_k\sim\mdist
\end{align}
where $\para\mid[\lat\sim\mdist]\vdash M:\PCFReal$ (slightly abusing notation in the trace type).

A classical choice for the step-size sequence is $\gamma_k\in\Theta(1/k)$, which satisfies the so-called \emph{Robbins-Monro} criterion:
\begin{align}
  \sum_{k\in\nat}\gamma_k&=\infty&\sum_{k\in\nat}\gamma^2_k&<\infty\label{eq:rm}
\end{align}

In this section we wish to establish the correctness of the SGD procedure applied to the smoothing \cref{eq:sgd}.

\subsection{Desiderata}
\label{sec:desi}
First, we ought to take a step back and observe that the optimisation problems we are trying to solve can be ill-defined due to a failure of integrability:
take $M\equiv(\lambda x\ldotp \PCFexp\,(x\infix\PCFmul x))\,\sample_\nd$: we have $\E_{z\sim\nd}[\sem M(z)]=\infty$, independently of parameters. Therefore, we aim to guarantee:
  \begin{enumerate}[label=(SGD\arabic*)]
    \setcounter{enumi}{-1}
    \setlength{\itemindent}{3em}
    \item\label{it:int} The optimisation problems (both smoothed and unsmoothed) are well-defined.
  \end{enumerate}

 Since $\E[\sema M(\para,\lat)]$ (and $\E[\sem M(\para,\lat)]$) may not be a convex function in the parameters $\para$, we cannot hope to always find global optima.
  We seek instead \emph{stationary points}, where the gradient w.r.t.~the parameters $\para$ vanishes.
The following results (whose proof is standard) provide sufficient conditions for the convergence of SGD to stationary points (see e.g.\ \cite{BT00} or \cite[Chapter~2]{B15}):
\begin{restatable}[Convergence]{proposition}{csgd}
  \label{prop:csgd}
  Suppose $(\gamma_k)_{k\in\nat}$ satisfies the Robbins-Monro criterion \cref{eq:rm} and
  $g(\para)\defeq\E_{\lat}[f(\para,\lat)]$
  is well-defined. If $\parasp\subseteq\Real^m$ satisfies

  \begin{enumerate}[label=(SGD\arabic*)]
    \setlength{\itemindent}{3em}
    \item\label{it:sgdunb} \textit{Unbiasedness:} $\nabla_\para g(\para)=\E_\lat[\nabla_\para f(\para,\lat)]$ for all $\para\in\parasp$
    \item\label{it:sgdlip} $g$ is $L$-Lipschitz smooth on $\parasp$ for some $L>0$:
    \[
    \|\nabla_\para g(\para)-\nabla_\para g(\para')\|\leq L\cdot \|\para-\para'\|\qquad \text{for all }\para,\para'\in\parasp
    \]
    \item\label{it:sgdbv} \textit{Bounded Variance:} $\sup_{\para\in\parasp}\E_{\lat}[\|\nabla_\para f_k(\para,\lat)\|^2]<\infty$
  \end{enumerate}
  then \emph{$\inf_{i\in\nat}\E[\|\nabla g(\para_i)\|^2]=0$} or $\para_i\not\in\parasp$ for some $i\in\nat$.
\end{restatable}
Unbiasedness \ref{it:sgdunb} requires commuting differentiation and integration. The validity of this operation can be established by the dominated convergence theorem \cite[Theorem 6.28]{K13}, see \cref{app:desi}. To be applicable the partial derivatives of $f$ w.r.t.\ the parameters need to be dominated uniformly by an integrable function.
Formally:
\begin{definition}
  Let $f:\parasp\times\Real^n\to\Real$ and $g:\Real^n\to\Real$. We say that $g$ \emph{uniformly dominates} $f$ if for all $(\para,\lat)\in\parasp\times\Real^n$, $|f(\para,\lat)|\leq g(\lat)$.
\end{definition}

Also note that for Lipschitz smoothness \ref{it:sgdlip} it suffices to uniformly bound the second-order partial derivatives.


In the remainder of this section we present two type systems which restrict the language to guarantee properties~\ref{it:int} to \ref{it:sgdbv}.

\subsection{Piecewise Polynomials and Distributions with Finite Moments}
\label{sec:poly}

As a first illustrative step we consider a type system $\polyj$, which restricts terms to (piecewise) polynomials,
and distributions with finite moments. Recall that a distribution $\dist$ has (all) \emph{finite moments} if for all $p\in\nat$, $\E_{s\sim\dist}[|s|^p]<\infty$.
Distributions with finite moments include the following commonly used distributions: normal, exponential, logistic and gamma distributions. A non-example is the Cauchy distribution, which famously does not even have an expectation.

\dw{terminology?}
\begin{definition}
  \changed[dw]{For a distribution $\dist$ with finite moments,}
  $f:\Real^n\to\Real$ has (all) \emph{finite moments} if for all $p\in\nat$, $\E_{\lat\sim\dist}[|f(\lat)|^p]<\infty$.
\end{definition}
Functions with finite moments have good closure properties:
\begin{restatable}{lemma}{absmom}
  \label{lem:absmom}
  If $f,g:\Real^n\to\Real$ have (all) finite moments so do $-f,f+g,f\cdot g$.
\end{restatable}
In particular, if a distribution has finite moments then polynomials do, too.
Consequently, intuitively, it is sufficient to simply (the details are explicitly spelled out in \cref{app:poly}):
\begin{enumerate}
  \item require that the distributions $\dist$ in the sample rule have finite moments:
  \begin{gather*}
    \infer[\dist\text{ has finite moments}]{\mid[\dist]\polyj\sample_\dist:\PCFReal}{}
  \end{gather*}
  \item remove the rules for $\PCFinv$, $\PCFexp$ and $\PCFlog$
   from the type system $\polyj$.
\end{enumerate} 


\subsubsection{Type Soundness I: Well-Definedness.}
Henceforth, we fix parameters $\theta_1:\iota_1,\ldots,\theta_m:\iota_m$.
Intuitively, it is pretty obvious that $\sem M$ is a piecewise polynomial whenever $\para\mid\trt\polyj M:\iota$. Nonetheless, we prove the property formally to illustrate our proof technique, a variant of logical relations, employed throughout the rest of the paper.

We define a slightly stronger logical predicate $\ppred^{(n)}_\tau$ on $\parasp\times\Real^n\to\sem\tau$, which allows us to obtain a uniform upper bound:
\begin{enumerate}
  \item $f\in\ppred^{(n)}_\iota$ if $f$ is uniformly dominated by a function with finite moments
  \item $f\in\ppred^{(n)}_{\tau_1\addtr\trt_3\to\tau_2}$ if for all $n_2\in\nat$ and $g\in\ppred^{(n+n_2)}_{\tau_1}$, $f\lrcomp g\in\ppred^{(n+n_2+|\trt_3|)}_{\tau_2}$
\end{enumerate}
where for $f:\parasp\times\Real^{n_1}\to\sem{\tau_1\addtr\trt_3\to\tau_2}$ and $g:\parasp\times\Real^{n_1+n_2}\to\sem{\tau_1}$ we define
\begin{align*}
  f\lrcomp g:\parasp\times\Real^{n_1+n_2+|\trt_3|}&\to\tau_2\\
  (\para,\tr_1\app\tr_2\app\tr_3)&\mapsto f(\para,\tr_1)(g(\para,\tr_1\app\tr_2),\tr_3)
\end{align*}
Intuitively, $g$ may depend on the samples in $\tr_2$ (in addition to $\tr_1$) and the function application may consume further samples $\tr_3$ (as determined by the trace type $\trt_3$).
By induction on safe types we prove the following result, which is important for conditionals:
\begin{lemma}
  \label{lem:ppredcond}
   If $f\in\ppred^{(n)}_\iota$ and $g,h\in\ppred^{(n)}_\safet$ then $[f(-)<0]\infix\cdot g+[f(-)\geq 0]\infix\cdot h\in\ppred^{(n)}_\safet$. \dw{this abuses notation. is it still understandable?}
\end{lemma}
\begin{proof}
  For base types it follows from \cref{lem:absmom}.
  Hence, suppose $\safet$ has the form $\safet_1\addtr[]\to\safet_2$. Let $n_2\in\nat$ and $x\in\ppred_{\safet_1}^{n+n_2}$. By definition, $(g\lrcomp x),(h\lrcomp x)\in\ppred_{\safet_2}^{(n+n_2)}$.
  Let $\widehat f$ be the extension (ignoring the additional samples) of $f$ to $\parasp\times\Real^{n+n_2}\to\Real$. It is easy to see that also $\widehat f\in\ppred_\iota^{(n+n_2)}$
  By the inductive hypothesis,
  \begin{align*}
    [\widehat f(-)<0]\cdot (g\lrcomp x)+[\widehat f(-)\geq 0]\cdot (h\lrcomp x)\in\ppred^{(n+n_2)}_{\safet_2}
  \end{align*} Finally, by definition,
  \begin{align*}
    ([f(-)<0]\cdot g+[f(-)\geq 0]\cdot h)\lrcomp x=[\widehat f(-)<0]\cdot (g\lrcomp x)+[\widehat f(-)\geq 0]\cdot (h\lrcomp x)
  \end{align*}
\end{proof}

\begin{assumption}
  \label{ass:comp}
  We assume that $\parasp\subseteq\sem{\iota_1}\times\cdots\times\sem{\iota_m}$ is compact.
\end{assumption}
\begin{restatable}[Fundamental]{lemma}{fundpolyf}
  \label{lem:fundpoly1}
  If $\para,x_1:\tau_1,\ldots,x_\ell:\tau_\ell\mid\trt\polyj M:\tau$, $n\in\nat$, $\xi_1\in\ppred^{(n)}_{\tau_1},\ldots,\xi_\ell\in\ppred^{(n)}_{\tau_\ell}$ then
  $\sem M\flcomp\langle\xi_1,\ldots,\xi_\ell\rangle\in\ppred^{(n+|\trt|)}_{\tau}$, where
  \begin{align*}
    \sem M\flcomp\langle\xi_1,\ldots,\xi_\ell\rangle:\parasp\times\Real^{n+|\trt|}&\to\sem\tau\\
    (\para,\tr\app\tr')&\mapsto \sem M((\para, \xi_1(\para,\tr),\ldots,\xi_\ell(\para,\tr)),\tr')
  \end{align*}
\end{restatable}
It is worth noting that, in contrast to more standard fundamental lemmas, here we need to capture the dependency of the free variables on some number $n$ of further samples.
E.g.\ in the context of $(\lambda x\ldotp x)\,\sample_\nd$ the subterm $x$ depends on a sample although this is not apparent if we consider $x$ in isolation.

\cref{lem:fundpoly1} is proven by structural induction (cf.~\cref{app:poly} for details). The most interesting cases include: parameters, primitive operations and conditionals. In the case for parameters we exploit the compactness of $\parasp$ (\cref{ass:comp}). For primitive operations we note that as a consequence of \cref{lem:absmom} each $\ppred_{\iota}^{(n)}$ is closed under negation\footnote{for $\iota=\PCFReal$}, addition and multiplication. Finally, for conditionals we exploit \cref{lem:absmom}.

\dw{conclude integrability and show Schwartz}

\subsubsection{Type Soundness II: Correctness of SGD.}
Next, we address the integrability for the \emph{smoothed} problem as well as \ref{it:sgdunb} to \ref{it:sgdbv}.
We establish that not only $\sema M$ but also its partial derivatives up to order 2 are uniformly dominated by functions with finite moments. For this to possibly hold we require:
\begin{assumption}
  \label{ass:sig1}
  For every $\eta>0$,
  \begin{align*}
    \sup_{x\in\Real}|\smooth(x)|&<\infty&
    \sup_{x\in\Real}|\smooth'(x)|&<\infty&
    \sup_{x\in\Real}|\smooth''(x)|&<\infty
  \end{align*}
\end{assumption}
Note that, for example, the logistic sigmoid satisfies \cref{ass:sig1}.

We can then prove a fundamental lemma similar to \cref{lem:fundpoly1}, \textit{mutatis mutandis}, using a logical predicate in $\VFr$. We stipulate
$f\in\qpred^{(n)}_{\iota}$ if its partial derivatives up to order 2 are uniformly dominated by a function with finite moments.
In addition to \cref{lem:absmom} we exploit standard rules for differentiation (such as the sum, product and chain rule) as well as \cref{ass:sig1}.  We conclude:

\begin{proposition}
  If $\para\mid\trt\polyj M:\PCFReal$ then the partial derivatives up to order 2 of $\sema M$ are uniformly dominated by a function with all finite moments.
\end{proposition}
Consequently, the Smoothed Optimisation \cref{p:optsm} is not only well-defined but, by the dominated convergence theorem \cite[Theorem 6.28]{K13}, the reparameterisation gradient estimator is unbiased. Furthermore, \ref{it:sgdunb} to \ref{it:sgdbv} are satisfied and SGD is correct.

\paragraph{Discussion.}
The type system $\polyj$ is simple yet guarantees correctness of SGD.
However, it is somewhat restrictive; in particular, it does not allow the expression of many ELBOs arising in variational inference directly as they often have the form of logarithms of exponential terms (cf.\ \cref{ex:varinf}).

\subsection{A Generic Type System with Annotations}
\label{sec:gentype}
Next, we present a generic type system with annotations.
In \cref{sec:int} we give an instantiation to make $\polyj$ more permissible and in \cref{sec:conv} we turn towards a different property: the uniform convergence of the smoothings.

Typing judgements have the form $\Gamma\mid\trt\genj M:\tau$, where ``$?$'' indicates the property we aim to establish, and we annotate base types. Thus, types are generated from
\begin{align*}
  &\textbf{trace types}&\trt&::=[s_1\sim\dist_1,\ldots,s_n\sim\dist_n]\\
  &\textbf{base types}&\iota&::=\PCFReal\mid\PCFRp\\
  &\textbf{safe types}&\safet&::=\iota^{\mathhighlight{\sann}}\mid\safet\addtr[]\to\safet\\
  &\textbf{types}&\tau&::=\iota^{\mathhighlight{\ann}}\mid\tau\addtr\trt\to\tau
\end{align*}
\dw{this may cause problems with name generation: $(\lambda f\ldotp f\,(f\,0))\,(\lambda x\ldotp\sample)$, but that's not so much of a problem (for soundness) because naming samples the same is safe}
Annotations are drawn from a set and
may possibly restricted for safe types. Secondly, the trace types are now annotated with variables, typically $\trt = [s_1\sim\dist_1,\ldots,s_n\sim\dist_n]$ where the variables $s_j$ are pairwise distinct.

For the subtyping relation we can constrain the annotations at the base type level (see \cref{fig:gensubtyping});
 the extension to higher types is accomplished as before.

The typing rules have the same form but they are extended with the annotations on base types and side conditions possibly constraining them.
For example, the rules for addition, exponentiation and sampling are modified as follows:
\begin{gather*}
  \infer[\mathhighlight{\text{(cond.\ Add)}}]{\mid[]\genj\PCFplus:\iota^{\ann_1}\to\iota^{\ann_2}\to\iota^{\ann}}{}\qquad
  \infer[\mathhighlight{\text{(cond.\ Exp)}}]{\mid[]\genj\PCFexp:\PCFReal^{\ann}\to\PCFRp^{\ann'}}{}\\
  \infer[\mathhighlight{\text{(cond.\ Sample)}}]{\mid[s_j\sim\dist]\genj\sample_\dist:\PCFReal^{\ann}}{}
\end{gather*}
The rules for subtyping, variables, abstractions and applications do not need to be changed at all but they use annotated types instead of the types of \cref{sec:basic type system}.
\begin{gather*}
  \infer[\Gamma\sqsubseteq_?\Gamma',\tau\sqsubseteq_?\tau']{\Gamma'\mid\trt\genj M:\tau'}{\Gamma\mid\trt\genj M:\tau}
  \qquad
  \infer{x:\tau \mid [] \genj x:\tau}{}\\
  \infer{\Gamma\mid[]\genj\lambda y\ldotp M:\tau_1\addtr\trt\to\tau_2}{\Gamma,y:\tau_1\mid\trt\genj M:\tau_2}\qquad
  \infer{\Gamma\mid\trt_1\app\trt_2\app\trt_3\genj M\,N:\tau_2}{\Gamma\mid\trt_2\genj M:\tau_1\addtr\trt_3\to\tau_2&\Gamma\mid\trt_1\genj N:\tau_1}
\end{gather*}
The full type system is presented in \cref{app:gentype}.

\begin{table}[t]
  \caption{Overview of type systems in this paper.}
  \label{tab:type}
  \centering
  \small
  \begin{tabular}{llcc}
    \toprule
    property & Section & judgement & annotation  \\
    \midrule
    totality & \cref{sec:basic type system} & $\vdash$ & -- \\
    \multirow{2}{*}{correctness SGD} & \cref{sec:poly}  & $\polyj$ & none/$\ast$  \\
    & \cref{sec:int} & $\intj$ & $0/1$\\
    uniform convergence & \cref{sec:gs} & $\gsj$ & $(\false,\dep)/(\true,\dep)$ \\
    \bottomrule
  \end{tabular}
\end{table}

$\polyj$ can be considered a special case of $\genj$ whereby we use the singleton $\ast$ as annotations, a contradictory side condition (such as $\mathrm{false}$) for the undesired primitives $\PCFinv$, $\PCFexp$ and $\PCFlog$, and use the side condition ``$\dist$ has finite moments'' for sample as above.

\cref{tab:type} provides an overview of the type systems of this paper and their purpose.
$\genj$ and its instantiations refine the basic type system of \cref{sec:basic type system} in the sense that if a term-in-context is provable in the annotated type system, then its erasure (i.e.~erasure of the annotations of base types and distributions) is provable in the basic type system.
This is straightforward to check.

\subsection{A More Permissible Type System}
\label{sec:int}
In this section we discuss another instantiation, $\intj$, of the generic type system system to guarantee \ref{it:int} to \ref{it:sgdbv}, which is more permissible than $\polyj$. In particular, we would like to support \cref{ex:varinf}, which uses logarithms and densities involving exponentials. Intuitively, we need to ensure that subterms involving $\PCFexp$ are ``neutralised'' by a corresponding $\PCFlog$. To achieve this we annotate base types with $0$ or $1$, ordered discretely. $0$ is the only annotation for safe base types and can be thought of as ``integrable''; $1$ denotes ``needs to be passed through log''.  More precisely, we constrain the typing rules such that if $\para\mid\trt\intj M:\iota^{(e)}$ then\footnote{using the convention $\log^0$ is the identity} $\log^e\circ\sem M$ and the partial derivatives of $\log^e\circ\sema M$ up to order 2 are uniformly dominated by a function with finite moments.


We subtype base types as follows: $\iota_1^{(e_1)}\sqsubseteq_{\mathrm{SGD}}\iota_2^{(e_2)}$ if $\iota_1\sqsubseteq\iota_2$ (as defined in \cref{fig:subtyping}) and $e_1=e_2$, or $\iota_1=\PCFRp=\iota_2$ and $e_1\leq e_2$. The second disjunct may come as a surprise but we ensure that terms of type $\PCFRp^{(0)}$ cannot depend on samples at all.

In \cref{fig:intjb} we list the most important rules; we relegate the full type system to \cref{app:int}.
\begin{figure}[t]
    \begin{gather*}
      \infer{\mid[]\intj\PCFexp:\PCFReal^{(0)}\to\PCFRp^{(1)}}{}\qquad
      \infer{\mid[]\intj\PCFlog:\PCFRp^{(e)}\to\PCFReal^{(0)}}{}\\
      \infer{\mid[]\intj\PCFplus:\iota^{(0)}\to\iota^{(0)}\to\iota^{(0)}}{}\qquad
      \infer{\mid[]\intj\PCFmul:\iota^{(e)}\to\iota^{(e)}\to\iota^{(e)}}{}\\
      \infer{\mid[]\intj\PCFmin:\PCFReal^{(0)}\to\PCFReal^{(0)}}{}\qquad
      \infer{\mid[]\intj\PCFinv:\PCFRp^{(e)}\to\PCFRp^{(e)}}{}\\
      \infer{\Gamma\mid\trt\app\trt'\app\trt''\intj\ifc LMN:\mathhighlight\safet}{\Gamma\mid\trt\intj L:\mathhighlight{\iota^{(0)}}&\Gamma\mid\trt'\intj M:\mathhighlight\safet&\Gamma\mid\trt''\intj N:\mathhighlight\safet}\\
      \infer[\dist\text{ has finite moments}]{\mid[s_j\sim\dist]\intj\sample_\dist:\PCFReal^{(0)}}{}
    \end{gather*}
    \caption{Excerpt of the typing rules (cf.~\cref{app:int}) for the correctness of SGD.}
    \label{fig:intjb}
\end{figure}
$\PCFexp$ and $\PCFlog$ increase and decrease the annotation respectively.
The rules for the primitive operations and conditionals are motivated by the closure properties of \cref{lem:absmom} and the elementary fact that
$\log\circ (f\cdot g)=(\log\circ f)+(\log\circ g)$ and $\log\circ (f^{-1})=-\log\circ f$
for $f,g:\parasp\times\Real^n\to\Real$.
\begin{example}
  $\theta:\PCFRp^{(0)}\mid[\nd,\nd]\intj\PCFlog\,(\theta^{-1}\cdot\PCFexp\,(\sample_\nd))+\sample_\nd:\PCFReal^{(0)}$
\end{example}
Note that the branches of conditionals need to have safe type, which rules out branches with type $\PCFReal^{(1)}$. This is because logarithms do not behave nicely when composed with addition as used in the smoothed interpretation of conditionals.

Besides, observe that in the rules for logarithm and inverses $e=0$ is allowed, which may come as a surprise\footnote{Recall that terms of type $\PCFRp^{(0)}$ cannot depend on samples.}.
This is e.g.\  necessary for the typability of the variational inference \cref{ex:varinf}:
\begin{example}[Typing for Variational Inference]
  \label{ex:varinftype}
It holds $\mid[]\vdash N:\PCFReal^{(0)}\to\PCFReal^{(0)}\to\PCFRp^{(0)}\to\PCFRp^{(1)}$ and $\theta:\PCFReal^{(0)}\mid[s_1\sim\nd]\vdash M:\PCFReal^{(0)}$.
\end{example}

\subsubsection{Type Soundness.}
To formally establish type soundness, we can use a logical predicate, which is very similar to the one in \cref{sec:poly} (N.B.\ the additional \cref{it:logpredc}):
in particular
$f\in\qpred^{(n)}_{\iota^{(e)}}$ if
  \begin{enumerate}
    \item partial derivatives of $\log^e\circ f$ up to order 2 are uniformly dominated by a function with finite moments
    \item\label{it:logpredc} if $\iota^{(e)}$ is $\PCFRp^{(0)}$ then $f$ is dominated by a positive constant function
  \end{enumerate}

Using this and a similar logical predicate for $\sem{(-)}$ we can show:
\begin{proposition}
  If $\theta_1:\iota^{(0)},\ldots,\theta_m:\iota_m^{(0)}\mid\trt\intj M:\iota^{(0)}$ then
  \begin{enumerate}
    \item all distributions in $\trt$ have finite moments
    \item $\sem M$ and for each $\eta>0$ the partial derivatives up to order 2 of $\sema M$ are uniformly dominated by a function with finite moments.
  \end{enumerate}
\end{proposition}
Consequently, again the Smoothed Optimisation \cref{p:optsm} is not only well-defined but by the dominated convergence theorem, the reparameterisation gradient estimator is unbiased. Furthermore, \ref{it:sgdunb} to \ref{it:sgdbv} are satisfied and SGD is correct.

%% file: sections/conv.tex

In the preceding section we have shown that SGD with the reparameterisation gradient can be employed to correctly (in the sense of \cref{prop:csgd}) solve the Smoothed Optimisation \cref{p:optsm} for any fixed accuracy coefficient.
However, \textit{a priori}, it is not clear how a solution of the Smoothed \cref{p:optsm} can help to solve the original \cref{p:opt}.

The following illustrates the potential for significant discrepancies:
\begin{example}
\label{ex:0g}
    Consider $M\equiv\ifc 0 {\theta\cdot\theta+\underline 1}{(\theta-\underline 1)\cdot(\theta-\underline 1)}$.
  Notice that the global minimum and the only stationary point of $\sema M$ is at $\theta=\frac 1 2$ regardless of $\eta>0$, where
  $\sema M(\frac 1 2)=\frac 3 4$.
  On the other hand $\sem M(\frac 1 2)=\frac 1 4$ and the global minimum of $\sem M$ is at $\theta=1$.
\end{example}

In this section we investigate under which conditions the smoothed objective function converges to the original objective function \emph{uniformly} in $\para\in\parasp$:
\begin{enumerate}[label=(Unif)]
  \setlength{\itemindent}{3em}
  \item\label{it:dunif} $\E_{\lat\sim\mdist}\left[\sema M(\para,\lat)\right]\unif\E_{\lat\sim\mdist}\left[\sem M(\para,\lat)\right]$ as $\eta\searrow 0$ for $\para\in\parasp$
\end{enumerate}
We design a type system guaranteeing this.

The practical significance of uniform convergence is that \emph{before} running SGD, for every error tolerance $\epsilon>0$ we can find an accuracy coefficient $\eta>0$ such that the difference between the smoothed and original objective function does not exceed $\epsilon$, in particular for $\para^*$ delivered by the SGD run for the $\eta$-smoothed problem.

\paragraph{Discussion of Restrictions.}
To rule out the pathology of \cref{ex:0g} we require that guards are non-0 almost everywhere.

Furthermore, as a consequence of the uniform limit theorem \cite{M99}, \ref{it:dunif} can only possibly hold if the \emph{expectation} $\E_{\lat\sim\mdist}\left[\sem M(\para,\lat)\right]$ is continuous (as a function of the parameters $\para$).
For a straightforward counterexample take $M\equiv\ifc\theta {\underline 0} {\underline 1}$, we have $\E_{\lat}[\sem M(\theta)]=[\theta\geq 0]$ which is discontinuous, let alone differentiable, at $\theta = 0$.
Our approach is to require that guards do not depend directly on parameters but they may do so, indirectly, via a diffeomorphic\footnote{\cref{ex:div} in \cref{app:conv} illustrates why it is \emph{not} sufficient to restrict the reparameterisation transform to \emph{bijections} (rather, we require it to be a diffeomorphism).} reparameterisation transform; see \cref{ex:affine transform}. We call such guards \emph{safe}.



In summary, our aim, intuitively, is to ensure that guards are the composition of a diffeomorphic transformation of the random samples (potentially depending on parameters) and a function which does not vanish almost everywhere.



\subsection{Type System for Guard Safety}
\label{sec:gs}

In order to enforce this requirement and to make the transformation more explicit, we introduce syntactic sugar, $\transt\dist\sd$, for applications of the form $\sd\,\sample_\dist$.
\begin{example}
  \label{ex:locscaletrans}
  As expressed in \cref{eq:locscale}, we can obtain samples from $\nd(\mu,\sigma^2)$ via $\transt\nd{(\lambda s\ldotp s\cdot\sigma+\mu)}$, which is syntactic sugar for the term $(\lambda s\ldotp s\cdot\sigma+\mu)\,\sample_\nd$.
\end{example}
We  propose another instance of the generic type system of \cref{sec:gentype}, $\gsj$, where we annotate base types by $\ann=(g,\dep)$, where $g\in\{\false,\true\}$ denotes whether we seek to establish guard safety and $\dep$ is a finite set of $s_j$ capturing possible dependencies on samples. We subtype base types as follows: $\iota_1^{(g_1,\dep_1)}\sqsubseteq_{\mathrm{unif}}\iota_2^{(g_2,\dep_2)}$ if $\iota_1\sqsubseteq\iota_2$ (as defined in \cref{fig:subtyping}),
$\dep_1\subseteq\dep_2$ and $g_1\preceq g_2$, where $\true\preceq\false$. This is motivated by the intuition that we can always drop\footnote{as long as it is not used in guards} guard safety and add more dependencies.

The rule for conditionals ensures that only safe guards are used.
The unary operations preserve variable dependencies and guard safety.
Parameters and constants are not guard safe and depend on no samples (see \cref{app:conv} for the full type system):
\begin{gather*}
  \infer{\Gamma\mid\trt\app\trt'\app\trt''\gsj\ifc LMN:\safet}{\Gamma\mid\trt\gsj L:\iota^{(\mathhighlight\true,\dep)}&\Gamma\mid\trt'\gsj M:\safet&\Gamma\mid\trt''\gsj N:\safet}\\
  \infer{\mid[]\gsj\PCFmin:\PCFReal^{(g,\dep)}\to\PCFReal^{(g,\dep)}}{}\\
  \infer{\theta_i:\iota^{(\false,\emptyset)}\mid[]\gsj\theta_i:\iota^{(\false,\emptyset)}}{}\qquad
  \infer[r\in\sem\iota]{\mid[]\gsj\underline r:\iota^{(\false,\emptyset)}}{}\\
  \infer[T\text{ diffeomorphic}]{\para\mid[s_j\sim\dist]\gsj\transt\dist T:\PCFReal^{(\true,\{s_j\})}}{\para\mid[]\gsj\sd:\PCFReal^{\ann}\to\PCFReal^{\ann}}
\end{gather*}
A term $\para\mid[]\gsj\sd:\PCFReal^{\ann}\to\PCFReal^{\ann}$ is diffeomorphic if $\sem T(\para,[])=\sema T(\para,[]) :\Real\to\Real$ is a diffeomorphism \dw{omitting $\Real^0$ due to trace type} for each $\para\in\parasp$, i.e.\ differentiable and bijective with differentiable inverse.

First, we can express affine transformations, in particular, the location-scale transformations as in \cref{ex:locscaletrans}:
\begin{example}[Location-Scale Transformation]
\label{ex:affine transform}
The term-in-context
\[
  \sigma:\PCFRp^{(\false,\emptyset)},\mu:\PCFReal^{(\false,\emptyset)} \mid [] \vdash \lambda s\ldotp\sigma\cdot s +\mu:\PCFReal^{(\false,\{s_1\})} \to \PCFReal^{(\false,\{s_1\})}
\]
is diffeomorphic.
(However for $\sigma:\mathhighlight{\PCFReal}^{(\false,\emptyset)}$ it is \emph{not} because it admits $\sigma=0$.)
Hence, the reparameterisation transform
\[
  G \equiv \sigma:\PCFRp^{(\false,\emptyset)},\mu:\PCFReal^{(\false,\emptyset)} \mid [s_1 : \mathcal D] \vdash  \transt \dist {(\lambda s.  s\cdot \sigma + \mu)} : \PCFReal^{(\true,\{ s_1 \})}
\]
which has $g$-flag $\true$, is admissible as a guard term.
Notice that $G$ depends on the parameters, $\sigma$ and $\mu$, \emph{indirectly} through a diffeomorphism, which is permitted by the type system.
\end{example}

If guard safety is sought to be established for the binary operations, we require that operands do not share dependencies on samples:
\begin{gather*}
  \infer[\circ\in\{+,\cdot\}]{\mid[]\gsj\PCFcirc:\iota^{(\false,\dep)}\to\iota^{(\false,\dep)}\to\iota^{(\mathhighlight\false,\dep)}}{}\\
  \infer[\circ\in\{+,\cdot\},\dep_1\cap\dep_2=\emptyset]{\mid[]\gsj\PCFcirc:\iota^{(\true,\dep_1)}\to\iota^{(\true,\dep_2)}\to\iota^{(\true,\dep_1\cup\dep_2)}}{}
\end{gather*}
This is designed to address:
\begin{example}[Non-Constant Guards]
  \label{ex:constant function cannot be used in guards}
  We have
  \(
    {}\mid[] \vdash (\lambda x . x + (- x)) : \PCFReal^{(\false,\{s_1\})} \to\PCFReal^{(\false,\{s_1\})},
  \)
  noting that we must use $g = \false$ for the $\PCFplus$ rule; and because $\PCFReal^{(\true,\{s_j\})} \sqsubseteq_{\mathrm{unif}} \PCFReal^{(\false,\{s_j\})}$, we have
  \[
    {}\mid[] \vdash (\lambda x . x + (\PCFmin x)) : \PCFReal^{(\true,\{s_1\})} \to\PCFReal^{(\false,\{s_1\})} .
  \]
  Now $\transt \dist {(\lambda y. y)}$ has type $\PCFReal^{(\true,\{s_1\})}$ with the $g$-flag necessarily set to $\true$;
  and so the term
  \[
    M \equiv \big(\lambda x. x+(- x)\big) \, \transt \dist {(\lambda y. y)}
  \]
  which denotes 0, has type $\PCFReal^{(\false,\{s_1\})}$, but \emph{not} $\PCFReal^{(\true,\{s_1\})}$.
  It follows that $M$ cannot be used in guards (notice the side condition of the rule for conditional), which is as desired: recall \cref{ex:0g}.
Similarly consider the term
\begin{align}
N &\equiv \big(\lambda x. (\lambda y \, z . \ifc {y + (- z)} {M_1} {M_2} ) \, x \, x \big)\notag
\\
&\qquad\qquad\qquad\qquad\qquad \,
(\transt \dist {(\lambda y. y)})
\label{eq:guard term 0 long}
\end{align}
When evaluated, the term ${y + (- z)}$ in the guard has denotation 0.
For the same reason as above, the term $N$ is not refinement typable.
\end{example}



The type system is however incomplete, in the sense that there are terms-in-context that satisfy the property \ref{it:dunif} but which are not typable.

\begin{example}[Incompleteness] The following term-in-context denotes the ``identity'':
\[
  {}\mid[] \vdash (\lambda x . (\underline{2} \cdot x) + (- x)) : \PCFReal^{(\true,\{s_1\})} \to \PCFReal^{(\false,\{s_1\})}
\]
but it does \emph{not} have type $\PCFReal^{(\true,\{s_1\})} \to\PCFReal^{(\true,\{s_1\})}$.
Then, using the same reasoning as \cref{ex:constant function cannot be used in guards}, the term
\[
  G \equiv (\lambda x. (\underline{2} \cdot x) + (- x)) \, (\transt \dist {(\lambda y. y)})
\]
has type $\PCFReal^{(\false,\{s_1\})}$, but \emph{not} $\PCFReal^{(\true,\{s_1\})}$, and so $\ifc G {\underline 0} {\underline 1}$ is not typable, even though $G$ can safely be used in guards.
\end{example}




\subsection{Type Soundness}
\label{sec:convsound}
Henceforth, we fix parameters $\theta_1:\iota^{(\false,\emptyset)}_1,\ldots,\theta_m:\iota^{(\false,\emptyset)}_m$.

Now, we address how to show property \ref{it:dunif}, i.e.\ that for $\para\mid\trt\gsj M:\iota^{(g,\dep)}$, the $\eta$-smoothed 
$\E[\sema M(\para,\lat)]$ converges uniformly for $\para\in\parasp$ as $\eta\searrow 0$.
For this to hold we clearly need to require that $\smooth$ has good (uniform) convergence properties (as far as the unavoidable discontinuity at $0$ allows for):
\begin{assumption}
  \label{ass:sig2}
  For every $\delta>0$, $\smooth\unif[(-)>0]$ on $(-\infty,-\delta)\cup(\delta,\infty)$.
\end{assumption}
Observe that in general even if $M$ is typable $\sema M$ does \emph{not} converge uniformly in both $\para$ and $\lat$ because $\sem M$ may still be discontinuous in $\lat$:
\begin{example}
  \label{ex:nconv}
  For $M\equiv\ifc{(\transt\nd{(\lambda s\ldotp s+\theta)})}{\underline 0}{\underline 1}$, $\sem M(\theta,s)=[s+\theta\geq 0]$, which is discontinuous, and $\sema M(\theta,s)=\smooth(s+\theta)$.
\end{example}

\iffalse
\begin{example}
  \label{ex:nconv}
  Observe that
  \begin{align*}
    f(\theta,z)\defeq\sem{\ifc{(\transt\nd(\lambda z\ldotp z\mathbin{\underline{+}}\theta))}{0}{1}}(\theta,z)=[z\mathbin{\underline{+}}\theta\geq 0]
  \end{align*} is discontinuous at $(\theta,-\theta)$.
  Therefore, despite the fact that
  \begin{align*}
    f_\eta(\theta,z)\defeq\sema{\ifc{(\transt\nd(\lambda z\ldotp z+\theta))}{0}{1}}(\theta,z)=\smooth(z\mathbin{\underline{+}}\theta)
  \end{align*}
  converges a.e.\ to $f$, the convergence cannot possibly be uniform.
\end{example}
\fi
However, if $\para\mid\Sigma\vdash M:\iota^{(g,\dep)}$ then $\sema M$ \emph{does} converge to $\sem M$ 
uniformly almost uniformly, i.e.,
uniformly in $\para\in\parasp$ and \emph{almost} uniformly in $\lat\in\Real^n$. Formally, we define:
\begin{definition}
  Let $f,f_\eta:\parasp\times \Real^n\to\Real$, $\mu$ be a \dw{what} measure on $\Real^n$.
  We say that $f_\eta$ \defn{converges uniformly almost uniformly} to $f$ (notation: $f_\eta\sau f$) if there exist sequences $(\delta_k)_{k\in\nat}$, $(\epsilon_k)_{k\in\nat}$ and $(\eta_k)_{k\in\nat}$ such that $\lim_{k \to \infty}\delta_k=0=\lim_{k \to \infty}\epsilon_k$; and for every $k\in\nat$ and $\para\in\parasp$ there exists $U\subseteq \Real^n$ such that
  \begin{enumerate}
    \item $\mu(U)<\delta_k$ and
    \item for every $0<\eta<\eta_k$ and $\lat\in\Real^n\setminus U$, $|f_\eta(\para,\lat)-f(\para,\lat)|<\epsilon_k$.
  \end{enumerate}
\end{definition}
If $f,f_\eta$ are independent of $\para$ this notion coincides with standard almost uniform convergence.
For $M$ from \cref{ex:nconv} $\sema M\sau\sem M$ holds although uniform convergence fails.

However, uniform almost uniform convergence entails uniform convergence of \emph{expectations}:
\begin{restatable}{lemma}{unisau}
  \label{lem:unisau}
  Let $f,f_\eta:\parasp\times\Real^n\to\Real$ have finite moments.

  If $f_\eta\sau f$ then $\E_{\lat\sim\mdist}[f_\eta(\para,\lat)]\unif\E_{\lat\sim\mdist}[f(\para,\lat)]$.
\end{restatable}

As a consequence, it suffices to establish $\sema M\sau\sem M$. We achieve this by positing an infinitary logical relation between sequences of morphisms in $\VFr$ (corresponding to the smoothings) and morphisms in $\qbs$ (corresponding to the measurable standard semantics). We then prove a Fundamental \cref{lem:fundconv} (details are in \cref{app:conv}).
Not surprisingly the case for conditionals is most interesting. This makes use of \cref{ass:sig2} and exploits that guards, for which the typing rules assert the guard safety flag to be $\true$, can only be $0$ at sets of measure $0$. We conclude:

\begin{theorem}
  If $\theta_1:\iota^{(\false,\emptyset)}_1,\ldots,\theta_m:\iota^{(\false,\emptyset)}_m\mid\trt\gsj M:\PCFReal^{(g,\dep)}$ then $\sema M\sau\sem M$.
   In particular, if $\sema M$ and $\sem M$ also have finite moments then
   \begin{align*}
     \E_{\lat\sim\mdist}[\sema M(\para,\lat)]&\unif\E_{\lat\sim\mdist}[\sem M(\para,\lat)]&\text{ as } \eta\searrow 0\text{ for }\para\in\parasp
   \end{align*}
\end{theorem}

We finally note that $\gsj$ can be made more permissible by adding syntactic sugar for $a$-fold (for $a\in\nat_{>0}$)
 addition $\PCFtimes M\equiv M\infix\PCFplus\cdots\infix\PCFplus M$ and multiplication $M\PCFpower\equiv M\infix\PCFmul\cdots\infix\PCFmul M$.  This admits more terms as guards, but safely (see \cref{fig:gsjt}).

%

%% file: sections/related.tex



\cite{LYY18} is both the starting point for our work and the most natural source for comparison.
They correct the (biased) reparameterisation gradient estimator for non-differentiable models by additional non-trivial \emph{boundary} terms.
They present an efficient method for \emph{affine} guards only.
Besides, they are not concerned with the \emph{convergence} of gradient-based optimisation procedures;
nor do they discuss how assumptions they make may be manifested in a programming language.

In the context of the reparameterisation gradient, \cite{MMT17} and \cite{JGP17} relax discrete random variables in a continuous way, effectively dealing with a specific class of discontinuous models.
\cite{Z81} use a similar smoothing for discontinuous optimisation but they do not consider a full programming language.

Motivated by guaranteeing absolute continuity (which is a necessary but not sufficient criterion for the correctness of e.g.\ variational inference), \cite{DBLP:journals/pacmpl/LewCSCM20} use an approach similar to our trace types to track the samples which are drawn.
They do not support standard conditionals but their ``work-around'' is also eager in the sense of combining the traces of both branches.
Besides, they do not support a full higher-order language, in which higher-order terms can draw samples. Thus, they do not need to consider \emph{function} types tracking the samples drawn during evaluation.

%% file: sections/experiments.tex

\dw{DISCUSS: variance vs.~approximation error}

\dw{constant vs.\ decreasing step size}


We evaluate our smoothed gradient estimator (\textsc{Smooth}) against the biased reparameterisation estimator (\textsc{Reparam}), the unbiased correction of it (\textsc{LYY18}) due to \cite{LYY18}, and the unbiased (\textsc{Score}) estimator \cite{RGB14,WW13,DBLP:conf/icml/MnihG14}.
The experimental setup is based on that of \cite{LYY18}.
The implementation is written in Python, using automatic differentiation (provided by the \texttt{jax} library) to implement each of the above estimators for an arbitrary probabilistic program.
For each estimator and model, we used the Adam \cite{DBLP:journals/corr/KingmaB14} optimiser for $10,000$ iterations using a learning rate of $0.001$, with the exception of \texttt{xornet} for which we used $0.01$.
The initial model parameters $\para_0$ were fixed for each model across all runs.
In each iteration, we used \changed[dw]{$N=16$} Monte Carlo samples from the gradient estimator.
For the \lyy estimator, a single subsample for the boundary term was used in each estimate.
For our smoothed estimator \changed[dw]{we use accuracy coefficients $\eta\in\{0.1,0.15,0.2\}$}.
Further details are discussed in \cref{app:expt-setup}.

\paragraph{Compilation for First-Order Programs.}

All our benchmarks are first-order. We compile a potentially discontinuous program to a smooth program (parameterised by $\smooth$) using the compatible closure of
\begin{align*}
  \ifc LMN\rightsquigarrow (\lambda w\ldotp \smooth(-w)\cdot M+\smooth(w)\cdot N)\,L
\end{align*}
Note that the size only increases linearly and that we avoid of an exponential blow-up by using abstractions rather than duplicating the guard $L$.

\paragraph{Models.}
\dw{models typable}
We include the models from \cite{LYY18}, an example from differential privacy \cite{DavidsonPilon15} and a neural network for which our main competitor, the estimator of \cite{LYY18}, is \emph{not} applicable (see \cref{app:models} for more details).

\begin{figure}[h!t]
  \centering
  \begin{subfigure}[h]{.48\textwidth}
    \includegraphics[width=\linewidth]{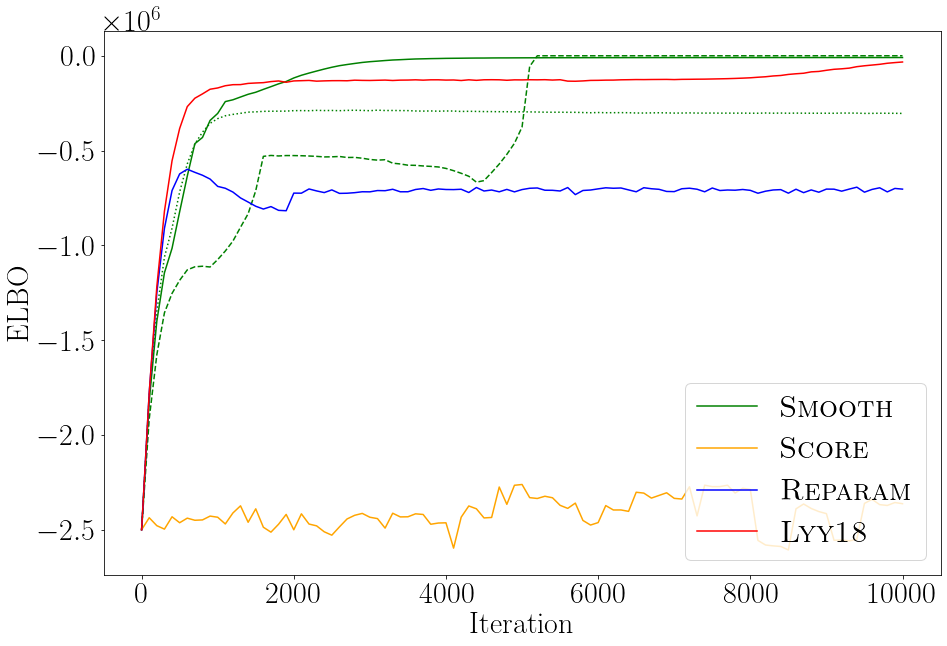}
    \caption{\texttt{temperature}}
    \label{fig:temperature-graph}
  \end{subfigure}
  \quad
  \begin{subfigure}[h]{.48\textwidth}
    \includegraphics[width=\linewidth]{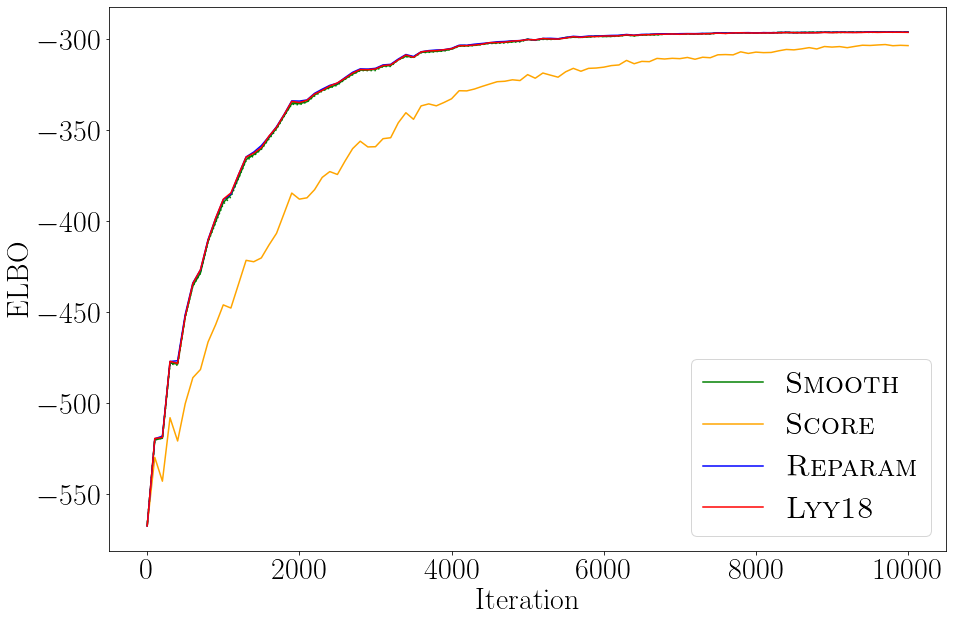}
    \caption{\texttt{textmsg}}
    \label{fig:textmsg-graph}
  \end{subfigure}
  \\
  \begin{subfigure}[h]{.48\textwidth}
    \includegraphics[width=\linewidth]{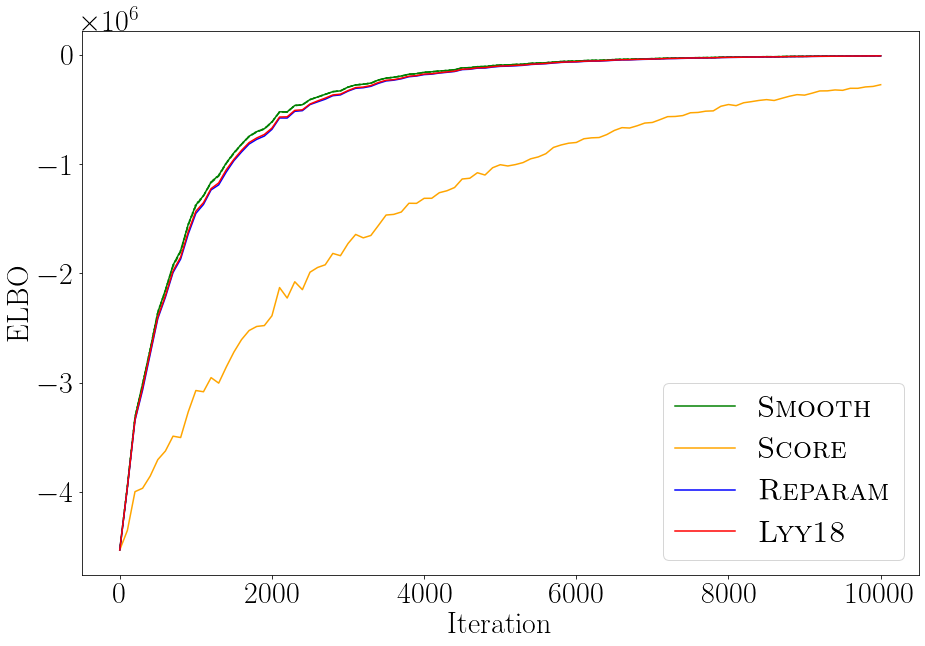}
    \caption{\texttt{influenza}}
    \label{fig:influenza-graph}
  \end{subfigure}
  \quad
  \begin{subfigure}[h]{.48\textwidth}
    \includegraphics[width=\linewidth]{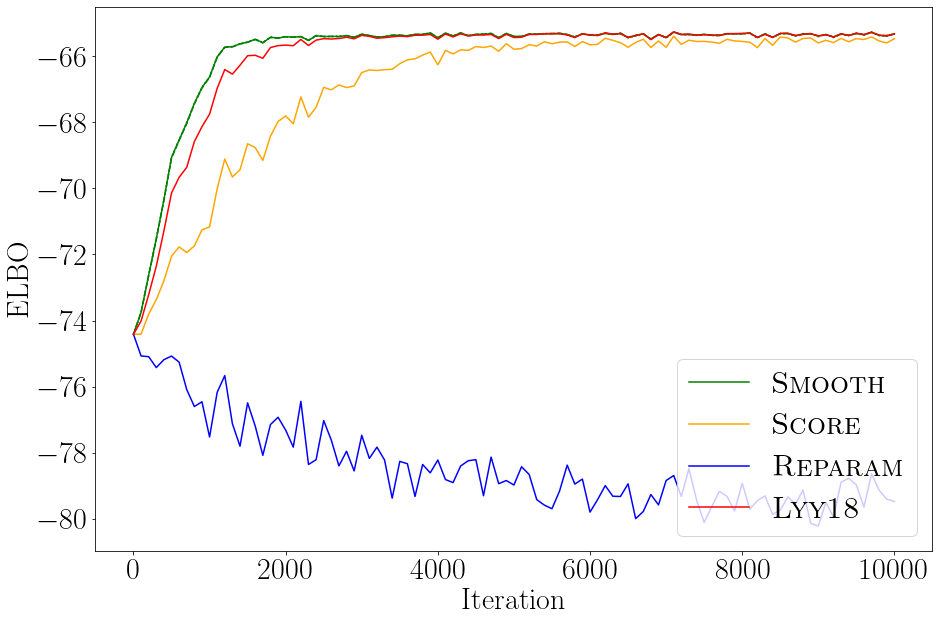}
    \caption{\texttt{cheating}}
    \label{fig:cheating-graph}
  \end{subfigure}
  \\[5pt]
  \begin{subfigure}[h]{.48\textwidth}
    \includegraphics[width=\linewidth]{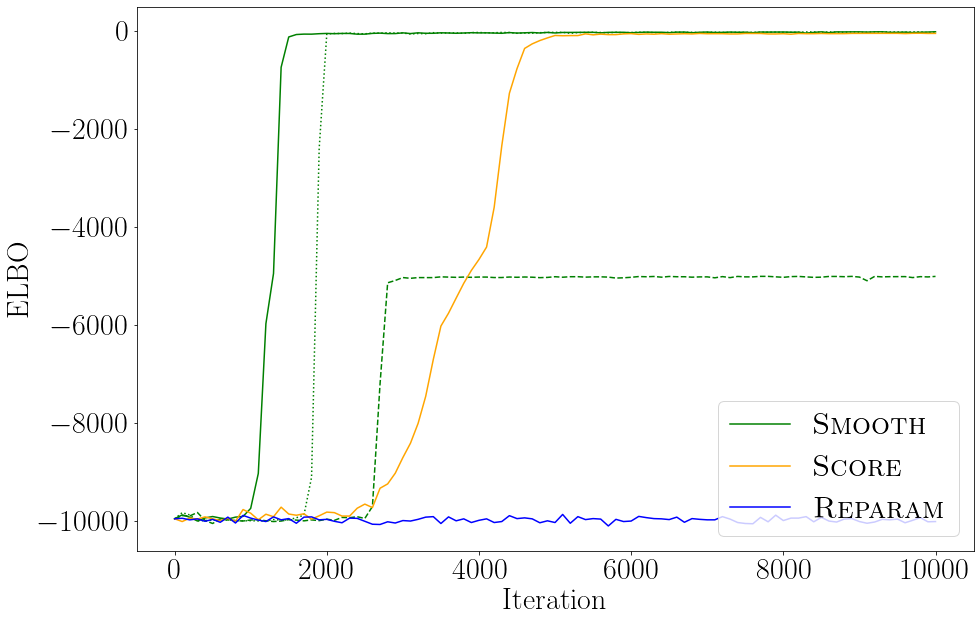}
    \caption{\texttt{xornet}}
    \label{fig:xornet-graph}
  \end{subfigure}
  \caption{ELBO trajectories for each model. A single colour is used for each
  estimator and the accuracy coefficient $\eta=0.1,0.15,0.2$ for \textsc{Smooth} is represented by dashed, solid and
  dotted lines respectively. \label{fig:estimator variance}}
\end{figure}

\subsection*{Analysis of Results}
We plot the ELBO trajectories in \cref{fig:estimator variance} and include data on the computational cost and variance in \cref{tab:var} in \cref{app:results}.


The ELBO graph for the \texttt{temperature} model in \cref{fig:temperature-graph} and the \texttt{cheating} model in \cref{fig:cheating-graph} shows that the \reparam estimator is biased, converging to suboptimal values when compared to the \nested and \lyy estimators.
For the \texttt{temperature} model we can also see from the graph and the data in \cref{tab:temperature} that the \score estimator exhibits extremely high variance, and does not converge.


Finally, the \texttt{xornet} model shows the difficulty of training step-function based neural nets.
The \lyy estimator is not applicable here since there are non-affine conditionals.
In \cref{fig:xornet-graph}, the \reparam estimator makes no progress while other estimators manage to converge to close to $0$ ELBO, showing that they learn a network that correctly classifies all points.
In particular, the \nested estimator converges the quickest. \dw{perhaps ommit if needed}

\textit{Summa summarum}, the results reveal where the \reparam estimator is biased and that the \nested estimator does not have the same limitation.
Where the \lyy estimator is defined, they converge to roughly the same objective value; and the smoothing approach is generalisable to more complex models such as neural networks with non-linear boundaries.
Our proposed \nested estimator has consistently significantly lower work-normalised variance, up to 3 orders of magnitude.


%% file: appendix/app-PL.tex

\subsection{Supplementary Materials for  \cref{sec:basic type system}}
\label{app:basict}

\unitt*
\begin{proof}[sketch]
  We define an equivalence relation $\approx$ on types by
  \begin{enumerate}
    \item $\iota\approx\iota'$
    \item $(\tau_1\addtr\trt \to \tau_2)
    \approx
    (\tau'_1\addtr\trt' \to \tau'_2)$ iff $\tau_1\approx\tau'_1$ implies $\trt=\trt'$ and $\tau_2\approx\tau'_2$
  \end{enumerate}
  Intuitively, two types are related by $\approx$ if for (inductively) related arguments they draw the same samples and again have related return types.
  We extend the relation to contexts: $\Gamma\approx\Gamma'$ if for all $x:\tau$ in $\Gamma$ and $x:\tau'$ in $\Gamma'$, $\tau\approx\tau'$.

  Then we show by induction that if $\Gamma\mid\trt\vdash M:\tau$, $\Gamma'\mid\trt'\vdash M:\tau'$ and $\Gamma\approx\Gamma'$ then $\trt=\trt'$ and $\tau\approx\tau'$. Finally, this strengthened statement allows us to prove the tricky case of the lemma: application.
\end{proof}

\subsection{Supplementary Materials for \cref{sec:densem}}
\label{app:densem}

Like measurable space $(X, \Sigma_X)$, a \emph{quasi Borel space} (QBS) is a pair $(X, M_X)$ where $X$ is a set;
but instead of axiomatising the measurable subsets $\Sigma_X$, QBS axiomatises the \emph{admissible random elements} $M_X$.
The set $M_X$, which is a collection of functions $\Real \to X$, must satisfy the following closure properties:
\begin{itemize}
\item if $\alpha \in M_X$ and $f : \Real \to \Real$ is measurable, then $\alpha \circ f \in M_X$
\item if $\alpha : \Real \to X$ is constant then $\alpha \in M_X$
\item given a countable partition of the reals $\Real = \biguplus_{i \in \nat} S_i$ where each $S_i$ is Borel, and $\{ \alpha_i \}_{i \in \nat} \subseteq M_X$, the function $r \mapsto \alpha_i(r)$ where $r \in S_i$ is in $M_X$.
\end{itemize}
The $\mathbf{QBS}$ morphisms $(X, M_X) \to (Y, M_Y)$ are functions $f : X \to Y$ such that $f \circ \alpha \in M_Y$ whenever $\alpha \in M_X$.

\dw{check}
\begin{lemma}[Substitution]
  \label{lem:substsem}
  Let $\Gamma,x:\tau'\mid\trt\vdash M:\tau$ and $\Gamma\mid[]\vdash N:\tau'$.

   Then
  $\sem M\big(\big(\gamma,\sem N(\gamma,[])\big),\lat\big)=\sem{M[N/x]}(\gamma,\lat)$.
\end{lemma}

\lo{Corollary - correctness of interpretation: Let $\Gamma \mid [] \vdash \lambda x. M: \tau' \addtr \Sigma \to \tau$ and $\Gamma \mid [] \vdash N:\tau'$. Then
\[
  \sem{\Gamma \mid \Sigma \vdash (\lambda x. M) \, N} = \sem{\Gamma \mid \Sigma \vdash M[N / x]}.
\]}

\subsection{Supplementary Materials for \cref{sec:opsem}}

\begin{figure}
\begin{framed}
  \begin{gather*}
    \infer{V\Downarrow_1^{[]} V}{}\qquad
    \infer{\sample_\dist\Downarrow^{[s]}_{\pdf_\dist(s)}\underline s}{}\\
    \infer[r<0]{\ifc LMN\Downarrow_{w_1\cdot w_2\cdot w_3}^{\tr_1\app\tr_2\app\tr_3} V}{L\Downarrow_{w_1}^{\tr_1}\underline r&M\Downarrow_{w_2}^{\tr_2} V&N\Downarrow_{w_3}^{\tr_3} V'}\\
    \infer[r\geq 0]{\ifc LMN\Downarrow_{w_1\cdot w_2\cdot w_3}^{\tr_1\app\tr_2\app\tr_3} V'}{L\Downarrow_{w_1}^{\tr_1}\underline r&M\Downarrow_{w_2}^{\tr_2} V&N\Downarrow_{w_3}^{\tr_3} V'}\\
    \infer[\circ\in\{+,\cdot\}]{M_1\infix\PCFcirc M_2\Downarrow_{w_1\cdot w_2}^{\tr_1\app\tr_2}\underline{r_1\circ r_2}}{M_1\Downarrow_{w_1}^{\tr_1}\underline{r_1}&M_2\Downarrow_{w_2}^{\tr_2}\underline{r_2}}\\
    \infer[\mathrm{op}\in\{-,{}^{-1},\exp,\log\}, r\in\dom(\mathrm{op})]{\underline{\mathrm{op}}\,M\Downarrow_w^\tr\underline{\mathrm{op}(r)}}{M\Downarrow_w^\tr\underline{r}}\\
    \qquad
    \infer{M\,N\Downarrow_{w_1\cdot w_2\cdot w_3}^{\tr_1\app\tr_2\app\tr_3} V}{M\Downarrow_{w_1}^{\tr_1}\lambda x\ldotp M'&N\Downarrow_{w_2}^{\tr_2} V'&M'[V'/x]\Downarrow_{w_3}^{\tr_3} V}
  \end{gather*}
  \caption{Operational big-step sampling-based semantics}
  \label{fig:ops}
\end{framed}
\end{figure}

The following can be verified by structural induction on $M$:
\begin{restatable}[Substitution]{lemma}{subst}
  \label{lem:substsimp}
  If $\Gamma,x:\tau'\mid\trt\vdash M:\tau$ and $\Gamma\mid[]\vdash N:\tau'$ then $\Gamma\mid\trt\vdash M[N/x]:\tau$.
\end{restatable}

Note that it may not necessarily hold that $\Gamma,x:\tau'\mid\trt\vdash M:\tau$ and $\Gamma\mid\trt'\vdash N:\tau'$ imply $\Gamma\mid\trt\app\trt'\vdash M[N/x]:\tau$.
Take $M \equiv x+x$ and $N \equiv \sample_\nd$.
Then note that
\begin{align*}
  x:\PCFReal\mid[]&\vdash M:\PCFReal&
  \mid[\nd]&\vdash N:\PCFReal&
  \mid[\nd,\nd]&\vdash M[N/x].
\end{align*}

\paragraph{Discussion.}
\cref{lem:substsimp} is a slightly stronger version of the usual substitution lemma for a CBV language: if $\Gamma, x : \tau' \mid \Sigma \vdash M : \tau$ and $\Gamma \mid \Sigma' \vdash V : \tau$ then $\Gamma \mid \Sigma \app  \Sigma' \vdash M[V / x] : \tau$; note that $\Sigma' = []$ necessarily, and we also have $\Gamma \mid \Sigma \app  \Sigma' \vdash (\lambda x.M) \, V : \tau$.
Consequently, subject reduction holds for CBV $\beta$-reduction.

%% file: appendix/app-smooth.tex

\begin{remark}
  \label{rem:gen}
  Suppose $\phi:X\to Y$ is a function and $(X,\mathcal C_X,\mathcal F_X)$ and $(Y,\mathcal C_Y,\mathcal C_X)$ are vector Fr\"olicher spaces, where the former is generated by $\mathcal C^0\subseteq\Set(\Real,X)$. Then $\phi$ is a morphism iff for all $f\in\mathcal F_Y$ and $c\in\mathcal C^0$, $f\circ\phi \circ c\in C^\infty(\Real,\Real)$ (i.e.\ it is not necessary to check $c\in\mathcal C_X\setminus\mathcal C^0$).
\end{remark}
(Note that $\mathcal C\subseteq\widetilde{\mathcal C}_X\subseteq\mathcal C_X$. Therefore, if $f:X\to\Real$ is such that for all $c\in\mathcal C_X\supseteq\mathcal C$, $f\circ c\in C^\infty(\Real,\Real)$ then $f\in\mathcal F_X$.)

\vfrccc*
\begin{proof}
  \begin{enumerate}
    \item Singleton vector spaces are terminal objects.
    \item Suppose $({X_1},\mathcal C_{X_1},\mathcal F_{X_1})$ and $({X_2},\mathcal C_{X_2},\mathcal F_{X_2})$ are vector Fr\"olicher spaces.
    Consider the vector Fr\"olicher space on $X_1\times X_2$ generated by $\{\langle c_1,c_2\rangle\mid c_1\in\mathcal C_{X_1}, c_2\in\mathcal C_{X_2}\}$.
    By construction
    $({X_1}\times {X_2},  \mathcal C_{{X_1}\times {X_2}},\mathcal F_{{X_1}\times {X_2}})$ is a vector Fr\"olicher space and $\pi_i:({X_1}\times {X_2},  \mathcal C_{{X_1}\times {X_2}},\mathcal F_{{X_1}\times {X_2}})\to ({X_i},  \mathcal C_{X_i},\mathcal F_{X_i})$ are morphisms.
    Now, suppose $Z$ and $f:Z\to {X_1}$ and $g:Z\to {X_2}$ are morphisms. Clearly, $h\defeq\langle f,g\rangle$ is the unique morphism $Z\to {X_1}\times {X_2}$ such that $\pi_1\circ h=f$ and $\pi_2\circ h=g$.

    \item Finally, suppose $(X,\mathcal C_X,\mathcal F_X)$ and $(Y,\mathcal C_Y,\mathcal F_Y)$ are vector Fr\"olicher spaces.
    Consider the vector Fr\"olicher space on the hom-set $\VFr(X,Y)$ generated by $\{c:\Real\to\VFr(X,Y)\mid ((r,x)\mapsto c(r)(x))\in \Fr(\Real\times X,Y)\}$.
    Define $\mathrm{eval}:\VFr(X, Y)\times X\to Y$ by $\mathrm{eval}(f,x)\defeq f(x)$. To see that this is a morphism by \cref{rem:gen} it suffices to consider $c_1:\Real\to C_{X\Rightarrow Y}$ such that $((r,x)\mapsto c_1(r)(x))\in\Fr(\Real\times X,Y)$, $c_2\in\mathcal C_X$ and $g\in\mathcal F_{Y}$.
    Note that
    \begin{align*}
      g\circ\mathrm{eval}\circ\langle c_1,c_2\rangle= g\circ\underbrace{ ((r,x)\mapsto c_1(r)(x))}_{\in \Fr(\Real\times X,Y)}\circ\underbrace{\langle\mathrm{id},c_2\rangle}_{\in\mathcal C_{\Real\times X}}
    \end{align*}
    which is in $C^\infty(\Real,\Real)$ by definition of morphisms.
    Clearly, this satisfies the required universal property.
  \end{enumerate}
\end{proof}

%% file: appendix/app-SGD.tex

\subsection{Supplementary Materials for \cref{sec:desi}}
\label{app:desi}

The following immediately follows from a well-known result about exchanging differentiation and integration, which is a consequence of the dominated convergence theorem \cite[Theorem 6.28]{K13}:
\begin{lemma}
  \label{lem:meas}
  Let $U \subseteq\Real$ be \changed[dw]{open}. Suppose $g\from\Real\times\Real^n\to\Real$ satisfies
  \begin{enumerate}
    \item for each $x\in\Real$, $\lat\mapsto g(x,\lat)$ is integrable
    \item $g$ is continuously differentiable everywhere
    \item there exists integrable $h\from\Real^n\to\Real$ such that for all $x\in U$ and $\lat\in\Real^n$, $|\frac{\partial g}{\partial x}(x,\lat)|\leq h(\lat)$.
  \end{enumerate}
  Then for all $x\in U$, $\frac\partial{\partial x}\int g(x,\lat)\diff\lat=\int\frac{\partial g}{\partial x}(x,\lat)\diff\lat$.
\end{lemma}

\begin{corollary}
  \label{cor:meas}
  Let $i\in\{1,\ldots,m\}$, $M>0$ and $U\defeq B_M(\mathbf 0)\subseteq\Real^m$ be the \changed[dw]{open $M$-ball}. Suppose $g\from\Real^m\times\Real^n\to\Real$ satisfies
  \begin{enumerate}
    \item for each $\mathbf x\in\Real^m$, $\lat\mapsto g(\mathbf x,\lat)$ is integrable
    \item $g$ is continuously differentiable everywhere
    \item there exists integrable $h\from\Real^n\to\Real$ such that for all $\mathbf x\in U$ and $\lat\in\Real^n$, $|\frac{\partial g}{\partial x_i}(\mathbf x,\lat)|\leq h(\lat)$.
  \end{enumerate}
  Then for all $\mathbf x\in U$, $\frac\partial{\partial x_i}\int g(\mathbf x,\lat)\diff\lat=\int\frac{\partial g}{\partial x_i}(\mathbf x,\lat)\diff\lat$.
\end{corollary}

\subsection{Supplementary Materials for \cref{sec:poly}}
\label{app:poly}

\absmom*
\begin{proof}
  For negation it is trivial.
  For addition it can be checked as follows:
  \begin{align*}
    \E[|(f+ g)(\lat)|^p]&\leq\E\left[|2f(\lat)|^p+ |2g(\lat)|^p\right]\\
    &\leq
    2^p\cdot\E\left[|f(\lat)|^p\right]+ 2^p\cdot\E\left[|g(\lat)|^p\right]<\infty
  \end{align*}
  For multiplication it follows from Cauchy-Schwarz:
  \begin{align*}
    \E[|(f\cdot g)(\lat)|^p]=\E\left[|f(\lat)|^p\cdot |g(\lat)|^p\right]\leq
    \sqrt{\E\left[|f(\lat)|^{2p}\right]\cdot\E\left[|g(\lat)|^{2p}\right] } <\infty
  \end{align*}
\end{proof}

\begin{figure}
\begin{framed}
  \begin{gather*}
    \infer[\Gamma\sqsubseteq_{\mathrm{poly}}\Gamma',
    \tau\sqsubseteq_{\mathrm{poly}}\tau']{\Gamma'\mid\trt\polyj M:\tau'}{\Gamma\mid\trt\polyj M:\tau}\qquad
    \infer{x:\tau\mid[]\polyj x:\tau}{}\\
    \infer[r\in\Real]{\mid[]\polyj \underline r:\PCFReal}{}\qquad
    \infer[r\in\Real_{>0}]{\mid[]\polyj \underline r:\PCFRp}{}\\
    \infer[\circ\in\{+,\cdot\}]{\mid[]\polyj \PCFcirc:\iota\to\iota\to\iota}{}\qquad
    \infer{\mid[]\polyj \underline-:\PCFReal\to\PCFReal}{}\\
    \infer{\Gamma\mid\trt\app\trt'\app\trt''\polyj\ifc LMN:\mathhighlight\safet}{\Gamma\mid\trt\polyj L:\PCFReal&\Gamma\mid\trt'\polyj M:\mathhighlight\safet&\Gamma\mid\trt''\polyj N:\mathhighlight\safet}\\
    \infer[\mathhighlight{\dist\text{ has finite moments}}]{\mid[s_j\sim\dist]\polyj\sample_{\mathcal\dist}:\PCFReal}{}\\
    \infer{\Gamma\mid[]\polyj\lambda y\ldotp M:\tau_1\addtr\trt\to\tau_2}{\Gamma,y:\tau_1\mid\trt\polyj M:\tau_2}\qquad
    \infer{\Gamma\mid\trt_1\app\trt_2\app\trt_3\polyj M\,N:\tau_2}{\Gamma\mid\trt_1\polyj M:\tau_1\addtr\trt_3\to\tau_2&\Gamma\mid\trt_2\polyj N:\tau_1}
  \end{gather*}
\end{framed}
\caption{Typing judgements for $\polyj$.}
\end{figure}

\fundpolyf*
\begin{proof}
  \dw{subtyping}
  We prove the claim by induction on $M$.
\begin{enumerate}
  \item For constants $\underline r$ and variables $x_i$ this is obvious; for parameters $\theta_i$ it is ensured by \cref{ass:comp}.
  \item $\sem{\sample_\dist}((),[s])=s$ clearly has finite moments because $\dist$ does.
  \item Next, to show $\sem\PCFplus\in\ppred^{(0)}_{\iota\to\iota\to\iota}$ (multiplication can be checked analogously) let $n_1,n_2\in\nat$, $f_1\in\ppred^{(n_1)}_{\iota}$, $f_2\in\ppred^{(n_1+n_2)}_{\iota}$. By definition $f_1$ and $f_2$ are uniformly dominated by some $g_1$ and $g_2$, respectively, with finite moments. By \cref{lem:absmom} $g_1+g_2$ has finite moments to and
  \begin{align*}
    |(\sem\PCFplus\lrcomp f_1\lrcomp f_2)(\para,\tr_1\app\tr_2)|&\leq |f_1(\para,\tr_1)|+|f_2(\para,\tr_1\app\tr_2)|\\
    &\leq g_1(\tr_1)+g_2(\tr_1\app\tr_2)
  \end{align*}
  \item The reasoning for $\PCFmin$ is straightforward and $\PCFinv$, $\PCFexp$ and $\PCFlog$ cannot occur.
  \item The claim for conditionals follows \cref{lem:ppredcond}.

  \item For applications it follows immediately from the inductive hypothesis and the definition.

  Suppose $\para,x_1:\tau_1,\ldots,\para,x_\ell:\tau_\ell\mid\trt_1\app\trt_2\app\trt_3:\tau_\ell\polyj M\,N:\tau$ because
  $\para,x_1:\tau_1,\ldots,\para,x_\ell:\tau_\ell\mid\trt_1:\tau_\ell\polyj M:\tau'\addtr\trt_3\to\tau$ and $\para,x_1:\tau_1,\ldots,\para,x_\ell:\tau_\ell\mid\trt_2:\tau_\ell\polyj M\,N:\tau'$. \dw{check order everywhere}

  Let $n\in\nat$ and $\xi_1\in\ppred^{(n)}_{\tau_1},\ldots,\xi_\ell\in\ppred^{(n)}_{\tau_\ell}$. By the inductive hypothesis,
  \begin{align*}
    \sem M\flcomp\langle\xi_1,\ldots,\xi_\ell\rangle&\in\ppred^{(n+|\trt_1|)}_{\tau'\addtr\trt_3\to\tau}&
    \sem N\flcomp\langle\xi_1,\ldots,\xi_\ell\rangle&\in\ppred^{(n+|\trt_1|+|\trt_2|)}_{\tau'}
  \end{align*}
  By definition of $\ppred^{(n+|\trt_1|)}_{\tau'\addtr\trt_3\to\tau}$,
  \begin{align*}
    (\sem M\flcomp\langle\xi_1,\ldots,\xi_\ell\rangle)\lrcomp(\sem N\flcomp\langle\xi_1,\ldots,\xi_\ell\rangle)\in\ppred^{(n+|\trt_1|+|\trt_2|+|\trt_3|)}_\tau
  \end{align*} and by definition of $\lrcomp$ and $\flcomp$, \dw{check}
  \begin{align*}
    (\sem M\flcomp\langle\xi_1,\ldots,\xi_\ell\rangle)\lrcomp(\sem N\flcomp\langle\xi_1,\ldots,\xi_\ell\rangle)=
    \sem {M\,N}\flcomp\langle\xi_1,\ldots,\xi_\ell\rangle
  \end{align*}
  \item For abstractions suppose $\para,x_1:\tau_1,\ldots,x_\ell:\tau_\ell\mid[]\polyj\lambda y\ldotp M:\tau\addtr\trt\to\tau'$ because $\para,x_1:\tau_1,\ldots,x_\ell:\tau_\ell,y:\tau\mid\trt\polyj M:\tau'$;
  let $n\in\nat$ and $\xi_1\in\ppred^{(n)}_{\tau_1},\ldots,\xi_\ell\in\ppred^{(n)}_{\tau_\ell}$.

  To show the claim, suppose $n_2\in\nat$ and $g\in\ppred_{\tau}^{(n+n_2)}$. By definition of the logical predicate we need to verify
  $\left(\sem M\flcomp\langle\xi_1,\ldots,\xi_\ell\rangle\right)\lrcomp g\in\ppred_{\tau'}^{(n+n_2+|\trt|)}$.

  Call $\widehat\xi_i$ the extension of $\xi_i$ to $\parasp\times\Real^{n+n_2}\to\Real$.
  By the inductive hypothesis,
  \begin{align*}
    \sem M\flcomp\langle\widehat\xi_1,\ldots,\widehat\xi_\ell,g\rangle\in\ppred_{\tau'}^{(n+n_2+|\trt|)}
  \end{align*}
  Finally it suffices to observe that \dw{check}
  \begin{align*}
    \left(\sem M\flcomp\langle\xi_1,\ldots,\xi_\ell\rangle\right)\lrcomp g=\sem M\flcomp\langle\widehat\xi_1,\ldots,\widehat\xi_\ell,g\rangle
  \end{align*}
\end{enumerate}
\end{proof}

\subsection{Supplementary Materials for \cref{sec:gentype}}
\label{app:gentype}
See \cref{fig:genty}.

\begin{figure}
  \begin{framed}

  \begin{subfigure}{\linewidth}
  \begin{gather*}
    \infer[\mathrm{(cond.\ subt.\ 1)}]{\iota^{\ann}\sqsubseteq_?\iota^{\ann'}}{}\qquad\infer[\mathrm{(cond.\ subt.\ 2)}]{\PCFRp^{\ann}\sqsubseteq_?\PCFReal^{\ann'}}{}\\
      \infer{(\tau_1\addtr\trt\to\tau_2)\sqsubseteq_?(\tau_1'\addtr\trt\to\tau_2')}{\tau'_1\sqsubseteq_?\tau_1&
      \tau_2\sqsubseteq_?\tau'_2}
  \end{gather*}
  \caption{Subtyping}
  \label{fig:gensubtyping}
\end{subfigure}

\begin{subfigure}{\linewidth}
\begin{gather*}
  \infer[\Gamma\sqsubseteq_?\Gamma',\tau\sqsubseteq_?\tau']{\Gamma'\mid\trt\genj M:\tau'}{\Gamma\mid\trt\genj M:\tau}
  \qquad
  \infer{x:\tau \mid [] \genj x:\tau}{}\\
  \infer{\Gamma\mid[]\genj\lambda y\ldotp M:\tau_1\addtr\trt\to\tau_2}{\Gamma,y:\tau_1\mid\trt\genj M:\tau_2}\qquad
  \infer{\Gamma\mid\trt_1\app\trt_2\app\trt_3\genj M\,N:\tau_2}{\Gamma\mid\trt_2\genj M:\tau_1\addtr\trt_3\to\tau_2&\Gamma\mid\trt_1\genj N:\tau_1}\\
  \infer[\mathrm{(cond.\ Para)}]{\theta_i:\iota^{\ann}\mid[]\genj\theta_i:\iota^{\ann}}{}\qquad
  \infer[\mathrm{(cond.\ Const)}, r\in\sem\iota]{\mid[]\genj\underline r:\iota^{\ann}}{}\\
  \infer[\mathrm{(cond.\ Add)}]{\mid[]\genj\PCFplus:\iota^{\ann_1}\to\iota^{\ann_2}\to\iota^{\ann}}{}\\
  \infer[\mathrm{(cond.\ Mul)}]{\mid[]\genj\PCFmul:\iota^{\ann_1}\to\iota^{\ann_2}\to\iota^{\ann}}{}\\
  \infer[\mathrm{(cond.\ Min)}]{\mid[]\genj\PCFmin:\PCFReal^{\ann}\to\PCFReal^{\ann}}{}\qquad
  \infer[\mathrm{(cond.\ Inv)}]{\mid[]\genj\PCFinv:\PCFRp^{\ann}\to\PCFRp^{\ann}}{}\\
  \infer[\mathrm{(cond.\ Exp)}]{\mid[]\genj\PCFexp:\PCFReal^{\ann}\to\PCFRp^{\ann'}}{}\qquad
  \infer[\mathrm{(cond.\ Log)}]{\mid[]\genj\PCFlog:\PCFRp^{\ann}\to\PCFReal^{\ann'}}{}\\
  \infer[\mathrm{(cond.\ If)}]{\Gamma\mid\trt\app\trt'\app\trt''\genj\ifc LMN:\safet}{\Gamma\mid\trt\genj L:\iota^{\ann}&\Gamma\mid\trt'\genj M:\safet&\Gamma\mid\trt''\genj N:\safet}\\
  \infer[\mathrm{(cond.\ Sample)}]{\mid[s_j\sim\dist]\genj\sample_\dist:\PCFReal^{\ann}}{}
\end{gather*}
\caption{Typing rules for $\genj$}
\end{subfigure}
\caption{Generic type system with annotations.}
\label{fig:genty}
\end{framed}
\end{figure}

\subsection{Supplementary Materials for \cref{sec:int}}
\label{app:int}
\begin{figure}
  \begin{framed}
    \begin{gather*}
      \infer[\Gamma\sqsubseteq_{\mathrm{SGD}}\Gamma',\tau\sqsubseteq_{\mathrm{SGD}}\tau']{\Gamma'\mid\trt\intj M:\tau'}{\Gamma\mid\trt\intj M:\tau}
      \qquad
      \infer{x:\tau \mid [] \intj x:\tau}{}\\
      \infer{\Gamma\mid[]\intj\lambda y\ldotp M:\tau_1\addtr\trt\to\tau_2}{\Gamma,y:\tau_1\mid\trt\intj M:\tau_2}\\
      \infer{\Gamma\mid\trt_1\app\trt_2\app\trt_3\intj M\,N:\tau_2}{\Gamma\mid\trt_2\intj M:\tau_1\addtr\trt_3\to\tau_2&\Gamma\mid\trt_1\intj N:\tau_1}\\
      \infer{\theta_i:\iota^{(0)}\mid[]\intj\theta_i:\iota^{(0)}}{}\qquad
      \infer[r\in\sem\iota]{\mid[]\intj\underline r:\iota^{(0)}}{}\\
      \infer{\mid[]\intj\PCFplus:\iota^{(0)}\to\iota^{(0)}\to\iota^{(0)}}{}\qquad
      \infer{\mid[]\intj\PCFmul:\iota^{(e)}\to\iota^{(e)}\to\iota^{(e)}}{}\\
      \infer{\mid[]\intj\PCFmin:\PCFReal^{(0)}\to\PCFReal^{(0)}}{}\qquad
      \infer{\mid[]\intj\PCFinv:\PCFRp^{(e)}\to\PCFRp^{(e)}}{}\\
      \infer{\mid[]\intj\PCFexp:\PCFReal^{(0)}\to\PCFRp^{(1)}}{}\qquad
      \infer{\mid[]\intj\PCFlog:\PCFRp^{(e)}\to\PCFReal^{(0)}}{}\\
      \infer{\Gamma\mid\trt\app\trt'\app\trt''\intj\ifc LMN:\safet}{\Gamma\mid\trt\intj L:\iota^{(0)}&\Gamma\mid\trt'\intj M:\safet&\Gamma\mid\trt''\intj N:\safet}\\
      \infer[\dist\text{ has finite moments}]{\mid[s_j\sim\dist]\intj\sample_\dist:\PCFReal^{(0)}}{}
    \end{gather*}
    \caption{Typing rules for $\intj$}
  \end{framed}
\end{figure}

We define the logical predicate $\qpred^{(n)}_\tau$ on $\parasp\times\Real^n\to\sem\tau$ in $\VFr$:
\begin{enumerate}
  \item $f\in\qpred^{(n)}_{\iota^{(e)}}$ if
  \begin{enumerate}
    \item partial derivatives of $\log^e\circ f$ up to order 2 are uniformly dominated by a function with finite moments
    \item if $\iota^{(e)}$ is $\PCFRp^{(0)}$ then $f$ is dominated by a positive constant function
  \end{enumerate}
  \item $f\in\ppred^{(n)}_{\tau_1\addtr\trt_3\to\tau_2}$ if for all $n_2\in\nat$ and $g\in\qpred^{(n+n_2)}_{\tau_1}$, $f\lrcomp g\in\qpred^{(n+n_2+|\trt_3|)}_{\tau_2}$.
\end{enumerate}

\begin{lemma}[Fundamental]
  If $\theta_1:\iota_1^{(0)},\ldots,\theta_m:\iota_m^{(0)},x_1:\tau_1,\ldots,x_\ell:\tau_\ell\mid\trt\intj M:\tau$, $n\in\nat$, $\xi_1\in\qpred^{(n)}_{\tau_1},\ldots,\xi_\ell\in\qpred^{(n)}_{\tau_\ell}$ then
  $\sema M\flcomp\langle\xi_1,\ldots,\xi_\ell\rangle\in\qpred^{(n+|\trt|)}_{\tau}$.
\end{lemma}
\begin{proof}
  Similar to \cref{lem:fundpoly1}, exploiting standard rules for logarithm and partial derivatives.
\end{proof}
\iffalse
\begin{proof}
\begin{enumerate}
  \item addition: exploit closure under addition

  Next, to show $\sema\PCFplus\in\qpred^{(0)}_{\iota^{(0)}\to\iota^{(0)}\to\iota^{(0)}}$ let $n_1,n_2\in\nat$, $f_1\in\qpred^{(n_1)}_{\iota^{(0)}}$, $f_2\in\qpred^{(n_1+n_2)}_{\iota^{(0)}}$. Let $g_1$ and $g_2$ be the respective dominating functinos. Call $\widehat g_2$ the extension of $g_2$ to $\Real^{n_1+n_2}\to\sem\iota$.
  Since $\widehat g_1$ and $g_2$ have finite absolute moments so does $\widehat g_1+g_2$ (\cref{lem:absmom}). \dw{reference} Furthermore clearly,
  \begin{align*}
    |(\sema\PCFplus\lrcomp f_1\lrcomp f_2)(\para,\tr_1\app\tr_2)|&=|f_1(\para,\tr_1)+f_2(\para,\tr_1\app\tr_2)|\\
    &\leq |f_1(\para,\tr_1)|+|f_2(\para,\tr_1\app\tr2)|\\
    &\leq g_1(\tr_1)+g_2(\tr_1\app\tr_2)\\
    &=(\widehat g_1+g_2)(\tr_1\app\tr_2)
  \end{align*}
  \dw{clearly constant if necessary}
  \item For multiplication we similarly exploit closure under addition and multiplication for $e=1$ and $e=0$, respectively (\cref{lem:absmom}). \dw{}
  \item primitive operations
  \item Suppose $\para,x_1:\tau_1,\ldots,x_\ell:\tau_\ell\mid\trt_1\app\trt_2\app\trt_3\intj\ifc LMN:\iota^{(e)}$ \dw{simplify rule accordingly, safe type} because $\para,x_1:\tau_1,\ldots,x_\ell:\tau_\ell\mid\trt_1\intj L:\PCFReal^{(0)}$, $\para,x_1:\tau_1,\ldots,x_\ell:\tau_\ell\mid\trt_2\intj M:\iota^{(0)}$ and $\para,x_1:\tau_1,\ldots,x_\ell:\tau_\ell\mid\trt_3\intj N:\iota^{(0)}$. Let $n\in\nat$ and $\xi_1\in\qpred_{\tau_1}^{(n)},\ldots,\xi_\ell\in\qpred_{\tau_\ell}^{(n)}$.

  By the inductive hypothesis there exist $f:\Real^{n+|\trt_1|}\to\Real$, $g:\Real^{n+|\trt_2|}\to\Real$ and $h:\Real^{n+|\trt_3|}\to\Real$ with absolute moments dominating $\sema L\flcomp\langle\xi_1,\ldots,\xi_\ell\rangle$, $\sema M\flcomp\langle\xi_1,\ldots,\xi_\ell\rangle$ and $\sema N\flcomp\langle\xi_1,\ldots,\xi_\ell\rangle$, respectively. Let $\widehat f,\widehat g,\widehat h$ be the respective extensions to $\Real^{n+|\trt_1|+|\trt_2|+|\trt_3|}\to\Real$. By \cref{lem:absmom} $g+h$ has finite moments. Besides,
  \begin{align*}
    &|(\sema {\ifc LMN}\flcomp\langle\xi_1,\ldots,\xi_\ell\rangle)(\para,\tr\app\tr_1\app\tr_2\app\tr_3)|\\
    &\leq|(\sema M\flcomp\langle\xi_1,\ldots,\xi_\ell\rangle)(\para,\tr\app\tr_2)|+|(\sema N\flcomp\langle\xi_1,\ldots,\xi_\ell\rangle)(\para,\tr\app\tr_3)|\\
    &\leq(\widehat g+\widehat h)(\tr\app\tr_1\app\tr_2\app\tr_3)
  \end{align*}

  \dw{constant if necessary}

  \dw{higher types}
  \item For abstractions \dw{similar}
  \item For applications it follows immediately from the inductive hypothesis and the definition.
  \item The reasoning for $-$ is similar.
  \item For $\mid[\dist]\intj\sample_\dist$ this is obvious because $\sema\sample=((),[s])\mapsto s$ is clearly a polynomial.   \item $\PCFinv$, $\PCFexp$ and $\PCFlog$ cannot occur.
\end{enumerate}
\end{proof}
\fi

%% file: appendix/app-conv.tex

\begin{example}[Divergence]
  \label{ex:div}
  Suppose $M\equiv\ifc{((\lambda z\ldotp z^3+\theta)\,\sample_\nd)}{\underline 0}{\underline 1}$.  Let $\phi_\theta(z)\defeq z^3+\theta$.
  Note that despite being bijective, $\phi_\theta\from\Real\to\Real$ is not a diffeomorphism because $\phi^{-1}_\theta(\alpha)=\sqrt[3]{\alpha-\theta}$ is not differentiable at $\alpha=\theta$.
  Then
  \begin{align*}
    \E_{z\sim\nd}[\sem M(\theta,z)]&=\int_{-\sqrt[3]{-\theta}}^\infty\nd(z\mid 0,1)\diff z\\
    \frac\partial{\partial\theta}\E_{z\sim\nd}[\sem M(\theta,z)]&=\frac 1 3\cdot \nd(-\sqrt[3]{-\theta}\mid 0,1)\cdot \theta^{-\frac 2 3}
  \end{align*}
  Therefore $\theta\mapsto\E_{z\sim\nd}[\sem M(\para,z)]$ is not differentiable at $0$. Besides, for $\theta=0$,
  \begin{align*}
    \E_{z\sim\nd}\left[\frac\partial{\partial\theta}\sema{M(\para,z)} \right]=\E_{z\sim\nd}\left[\sigma_\eta'(z^3)\right]\to\infty
  \end{align*}
\end{example}

\begin{figure}
  \begin{framed}
    \begin{gather*}
      \infer[\Gamma\sqsubseteq_{\mathrm{unif}}\Gamma',\tau\sqsubseteq_{\mathrm{unif}}\tau']{\Gamma'\mid\trt\gsj M:\tau'}{\Gamma\mid\trt\gsj M:\tau}
      \qquad
      \infer{x:\tau \mid [] \gsj x:\tau}{}\\
      \infer{\Gamma\mid[]\gsj\lambda y\ldotp M:\tau_1\addtr\trt\to\tau_2}{\Gamma,y:\tau_1\mid\trt\gsj M:\tau_2}\qquad
      \infer{\Gamma\mid\trt_1\app\trt_2\app\trt_3\gsj M\,N:\tau_2}{\Gamma\mid\trt_2\gsj M:\tau_1\addtr\trt_3\to\tau_2&\Gamma\mid\trt_1\gsj N:\tau_1}\\
      \infer{\theta_i:\iota^{(\false,\emptyset)}\mid[]\gsj\theta_i:\iota^{(\false,\emptyset)}}{}\qquad
      \infer[r\in\sem\iota]{\mid[]\gsj\underline r:\iota^{(\false,\emptyset)}}{}\\
      \infer[\circ\in\{+,\cdot\}]{\mid[]\gsj\PCFcirc:\iota^{(\false,\dep)}\to\iota^{(\false,\dep)}\to\iota^{(\false,\dep)}}{}\\
      \infer[\circ\in\{+,\cdot\},\dep_1\cap\dep_2=\emptyset]{\mid[]\gsj\PCFcirc:\iota^{(\true,\dep_1)}\to\iota^{(\true,\dep_2)}\to\iota^{(\true,\dep_1\cup\dep_2)}}{}\\
      \infer{\mid[]\gsj\PCFmin:\PCFReal^{(g,\dep)}\to\PCFReal^{(g,\dep)}}{}\qquad
      \infer{\mid[]\gsj\PCFinv:\PCFRp^{(g,\dep)}\to\PCFRp^{(g,\dep)}}{}\\
      \infer{\mid[]\gsj\PCFexp:\PCFReal^{(g,\dep)}\to\PCFRp^{(g,\dep)}}{}\qquad
      \infer{\mid[]\gsj\PCFlog:\PCFRp^{(g,\dep)}\to\PCFReal^{(g,\dep)}}{}\\
      \infer{\Gamma\mid\trt\app\trt'\app\trt''\gsj\ifc LMN:\safet}{\Gamma\mid\trt\gsj L:\iota^{(\true,\dep)}&\Gamma\mid\trt'\gsj M:\safet&\Gamma\mid\trt''\gsj N:\safet}\\
      \infer{\mid[s_j\sim\dist]\gsj\sample_\dist:\PCFReal^{(\true,\{s_j\})}}{}\\
      \infer[T\text{ diffeomorphic}]{\para\mid[s_j\sim\dist]\gsj\transt\dist T:\PCFReal^{(\true,\{s_j\})}}{\para\mid[]\gsj\sd:\PCFReal^{\ann}\to\PCFReal^{\ann}}\\
      \infer[a\in\nat_{>0}]{\Gamma\mid[]\gsj\infix\PCFtimes M:\iota^{(\true,\dep)}}{\Gamma\mid[]\gsj M:\iota^{(\true,\dep)}}\qquad
      \infer[a\in\nat_{>0}]{\Gamma\mid[]\gsj M\infix\PCFpower:\iota^{(\true,\dep)}}{\Gamma\mid[]\gsj M:\iota^{(\true,\dep)}}
    \end{gather*}
    \caption{Typing rules for $\gsj$.}
    \label{fig:gsjt}
  \end{framed}
\end{figure}

\subsection{Properties of Uniform Almost Uniform Convergence}


Let $\mu(U)=\E_{\lat\sim\mdist}[[\lat\in U]]$, where $\mdist$ has finite moments and $\rep_\para$ be a diffeomorphism. We continue assuming compactness of $\parasp$.

\begin{lemma}
  \label{lem:balldim}
  $\lim_{k\in\nat}\sup_{\para\in\parasp}\mu(\rep_\para^{-1}(\Real^n\setminus\ball k))=0$
\end{lemma}
\begin{proof}
  Let $\lat_0\in\Real^n$ be arbitrary.
  \renewcommand\ball[1]{\mathbf B_{#1}(\lat_0)}

  \begin{align*}
    \delta_*^{(i)}&\defeq\sup_{\para\in\parasp}|\rep_\para^{(i)}(\lat_0)|\\
    d_k^{(i)}&\defeq\sup_{\para\in\parasp}\sup_{\lat\in\ball k}\|\nabla_\para\rep^{(i)}_\para(\lat)\|
  \end{align*}
  thus if $\lat\in\ball k$,
  \begin{align*}
    |\rep_\para(\lat)^{(i)}|&\leq \|\rep^{(i)}_\para(0)\|+\langle\nabla\rep^{(i)}_\para(\zeta),x\rangle\\
    &\leq \delta_*^{i}+\|\nabla\rep^{(i)}_\para(\zeta)\|\cdot\|x\|\\
    &\leq \delta_*^{i}+ d_k^{(i)}\cdot k
  \end{align*}
  Let
  \begin{align*}
    \delta_k^{(i)}&\defeq \delta_*^{i}+ d_k^{(i)}\cdot k&
    \delta_k&\defeq\sqrt n\cdot\max_{i=1}^n\delta_k^{(i)}
  \end{align*}
  Note that for $\lat\in\Real^n$,
  if $\|\rep_\para(\lat)\|>\delta_k$ then $|\rep_\para^{(i)}(\lat)|>\delta_r^{(i)}$ for some $1\leq i\leq n$ and thus $\lat\in\Real^n\setminus\ball k$.
  As a consequence, $\rep_\para^{-1}(\Real^n\setminus\mathbf B_{\delta_k}(\mathbf 0))\subseteq\Real^n\setminus\ball k$.

  Now, it suffices to observe that $\lim_k\mu(\Real^n\setminus\ball k)=0$. \dw{reason in opposite direction}
\end{proof}

\begin{lemma}
  For each $k\in\nat$ there exists $c>0$ such that $\mu(\rep^{-1}_\para(U\cap\ball k))\leq c\cdot\Leb(U\cap\ball k)$.
\end{lemma}
\begin{proof}
  Let $f:\Real^n\to\Real$ be the density of $\mu$ \dw{details}. Then
  \begin{align*}
    \mu(\rep^{-1}_\para(U\cap\ball k))&=\int_{\rep^{-1}_\para(U\cap\ball k)} f(\lat)\diff\lat\\
    &=\int_{U\cap\ball k} f(\rep^{-1}_\para(\repx))\cdot|\det\Jac\rep^{-1}_\para(\repx)|\diff\repx
  \end{align*}
  \dw{there exists upper bound by continuity of $\rep_\para^{-1}$}
\end{proof}

\begin{lemma}
  \label{lem:ausig}
  Suppose $f_\eta\circ\rep_{(-)}(-)\sau f\circ\rep_{(-)}(-)$ and $f\neq 0$ a.e.

  Then $\smooth\circ f_\eta\circ\rep_{(-)}(-)\sau [f(\rep_{(-)}(-))>0]$.
\end{lemma}
\begin{proof}
  Let $\delta_k$, $\epsilon_k$ and $\eta_k$ be witnesses for $f_\eta\circ\rep_{(-)}(-)\sau f\circ\rep_{(-)}(-)$.

  For $i\in\nat$ define $V_i\defeq\{\repx\in\Real^n\mid|f(\repx)|<\frac 1 i\}$.
  For every $k\in\nat$ there exists $i_k\in\nat$ such that $\Leb(V_{i_k}\cap\ball k)<\frac 1 k$.
    (This is because $\Leb((-)\cap\ball k)$ is a finite measure and $\cap_{i\in\nat} V_{i}\subseteq f^{-1}(0)$ and $f\neq 0$ a.e.)

  Furthermore, for $k\in\nat$ let $K_k\in\nat$ be such that $\epsilon_{K_k}<\frac 1{2{i_k}}$.
  By \cref{ass:sig2} there exists $0<\eta_k'<\eta_{K_k}$ such that for all $0<\eta<\eta'_k$ and $y>\frac 1{2i_k}$, $\smooth(-y)<\frac 1 k$ and $\smooth(y)>1-\frac 1 k$.
  We also define
  \begin{align*}
    \delta'_k&\defeq\delta_{K_k}+\sup_{\para\in\parasp}\mu(\rep_\para^{-1}(\Real^n\setminus\ball k))+\frac 1 k&
    \epsilon'_k&\defeq\frac 1 k
  \end{align*}
  By \cref{lem:balldim}, $\lim\delta'_k=0=\lim\epsilon'_k$.

  Now, suppose $\para\in\parasp$ and $k\in\nat$.

  Define $U'\defeq U_{K_k}\cup\rep_\para^{-1}(V_{i_k})$ where $U_{K_k}\subseteq\Real^n$ is the corresponding set for $[f(\rep_{(-)}(-))>0]$, $\para$ and $K_k$.
  It holds
  \begin{align*}
    \mu(U')&\leq
    \mu(U_{K_k})+\mu(\rep_\para^{-1}(\Real^n\setminus\ball{k}))+\mu(\rep_\para^{-1}(V_{i_k}\cap\ball{k}))\\
    &\leq\mu(U_{K_k})+\mu(\rep_\para^{-1}(\Real^n\setminus\ball{k}))+c\cdot\Leb(V_{i_k}\cap\ball{k})\\
    &\leq\delta'_k
  \end{align*}
  \dw{$c$, make sure it works}
  Besides,
  for $0<\eta<\eta_k'$ and $\lat\in\Real^n\setminus U'$,
  $|f_\eta(\rep_\para(\lat))-f(\rep_\para(\lat))|<\frac 1{2i_k}$ and $|f(\rep_\para(\lat))|\geq\frac 1{i_k}$ thus
  $|f_\eta(\rep_\para(\lat))|>\frac 1{2{i_k}}$.

  Consequently, $|\smooth(f_\eta(\rep_\para(\lat)))-[f(\rep_\para(\lat))>0]|<\frac 1 k$.
\end{proof}

\begin{lemma}
  \label{lem:aucomp}
  If $f:U_1\times U_2\to\Real$ (for open and \changed[dw]{connected} $U_1,U_2\subseteq\Real$) is continuously differentiable and $g_\eta\sau g:\parasp\times\Real^n\to U_1$ and $h_\eta\sau h:\parasp\times\Real^n\to U_2$, \changed[dw]{$g,h$ are also bounded on bounded subsets of $\Real^n$} then $f\circ\langle g_\eta,h_\eta\rangle\sau f\circ\langle g,h\rangle:\parasp\times\Real^n\to\Real$. \dw{is also bounded on bounded subsets of $\Real^n$}.
\end{lemma}
\begin{proof}
  First, note that $f\circ\langle g,h\rangle$ is bounded on bounded subsets of $\Real^n$ because $f$ is continuously differentiable and $g$ and $h$ also satisfies this property.

  Let $\delta_k^{(i)}$, $\epsilon_k^{(i)}$ and $\eta_k^{(i)}$ ($i\in\{1,2\}$) be witnesses for $g_\eta\sau g$ and $h_\eta\sau h$.
  W.l.o.g.\ all $\epsilon_k^{(i)}\leq 1$.
  Observe that for $k\in\nat$,
  \begin{align*}
    M_k&\defeq \sup_{(\para,\lat)\in\parasp\times\ball k} \|(g(\para,\lat),h(\para,\lat))\|+\sqrt 2<\infty
  \end{align*}
  because
  $g(\parasp\times\ball k)$ and $h(\parasp\times\ball k)$ are bounded by assumption (also \cref{ass:comp})
  and therefore
  \begin{align*}
    d_k&\defeq\sup_{(x,y)\in M_k\cap(U_1\times U_2)}\left\|\nabla f(x,y)\right\|<\infty
  \end{align*}
   is well-defined.
  For $k\in\nat$ there exists $K_k\geq k$ such that each
  $\sqrt 2\cdot d_k\cdot\epsilon_{K_k}^{(i)}<\frac 1 k$.

  Define
  \begin{align*}
    \delta_k&\defeq\mu(\Real^n\setminus\ball k)+\delta^{(1)}_{K_k}+\delta^{(2)}_{K_k}&
    \epsilon_k&\defeq\frac 1 k&
    \eta_k&\defeq\min\{\eta_{K_k}^{(1)},\eta_{K_k}^{(2)}\}
  \end{align*}
  Note that by \cref{lem:balldim}, $\lim\delta_k=0=\lim\epsilon_k$.

  Let $\para\in\parasp$ and $k\in\nat$.
  Let $V\defeq (\Real^n\setminus\ball k)\cup V^{(1)}\cup V^{(2)}$, where $V^{(1)}$ (and $V^{(2)}$, respectively) are the sets for $g$ (and $h$, respectively), $\para$ and $K_k$.
  Note that $\mu(V)\leq\delta_k$. Besides for every $0<\eta<\eta_k$ and $\lat\in\Real^n\setminus V$,
  $|g_\eta(\para,\lat)|\leq|g(\para,\lat)|+\epsilon_{K_k}^{(1)}\leq|g(\para,\lat)|+1$ (similarly for $h$). Hence,
  every point between $(g_\eta(\para,\lat),h_\eta(\para,\lat))$ and $(g(\para,\lat),h(\para,\lat))$ is in $\ball {M_k}\cap(U_1\times U_2)$ and therefore by the mean value theorem,
  \begin{align*}
    &\left|f(g_\eta(\para,\lat),h_\eta(\para,\lat))-f(g(\para,\lat),h(\para,\lat))\right|\\
    &\leq
    \sup_{\zeta\in \ball {M_k}\cap(U_1\times U_2)}\left|\left\langle\nabla f(\zeta),\left(g_\eta(\para,\lat)-g(\para,\lat),h_\eta(\para,\lat)-h(\para,\lat)\right)\right\rangle\right|\\
    &\leq\sup_{\zeta\in \ball {M_k}\cap(U_1\times U_2)}\|\nabla f(\zeta)\|\cdot
    \left\|\left(g_\eta(\para,\lat)-g(\para,\lat),h_\eta(\para,\lat)-h(\para,\lat)\right)\right\|\\
    &< d_k\cdot\sqrt{2}\cdot\max\{\epsilon_{K_k}^{(1)},\epsilon_{K_k}^{(2)}\}\\
    &\leq\epsilon_k
    &&
  \end{align*}
  using the Cauchy–Schwarz inequality in the second step.
\end{proof}



\unisau*
\begin{proof}
  It suffices to show the uniform convergence of $\E_{\lat\sim\mdist}[|f_\eta(\para,\lat)-f(\para,\lat)|]$ to $0$.

  By assumption there exists $M>0$ such that $\E_{\lat\sim\mdist}\left[|f_\eta(\para,\lat)-f(\para,\lat)|^2\right]<M$ for all $\eta>0$ and $\para\in\parasp$. \dw{needs slightly stronger or argument}

  Let $\epsilon>0$. By uniform almost uniform convergence of $f_\eta$ to $f$ there exists $k$ such that $\delta_k\cdot M,\epsilon_k<\frac\epsilon 2$.

  Suppose $\para\in\parasp$ and $\eta<\eta_k$.
  Let $U\subseteq\Real^n$ be the witness for almost uniform convergence of $f_\eta$, $k$ and $\para$. In particular,
  $\E_{\lat\sim\mdist}[[\lat\in U]]\cdot M<\delta_k\cdot M<\frac\epsilon 2$ and for every $\lat\in\Real^n\setminus U$, $|f_\eta(\para,\lat)-f(\para,\lat)|<\epsilon_k<\frac\epsilon 2$.

  \begin{align*}
    &\E_{\lat\sim\mdist}[|f_\eta(\para,\lat)-f(\para,\lat)|]\\
    &\leq\E_{\lat\sim\mdist}\left[[\lat\in U]\cdot|f_\eta(\para,\lat)-f(\para,\lat)|\right]+\E_{\lat\sim\mdist}\left[[\lat\in\Real^n\setminus U]\cdot|f_\eta(\para,\lat)-f(\para,\lat)|\right]\\
    &\leq\E_{\lat\sim\mdist}\left[[\lat\in U]\right]\cdot\E_{\lat\sim\mdist}\left[|f_\eta(\para,\lat)-f(\para,\lat)|^2\right]+\E_{\lat\sim\mdist}\left[[\lat\in\Real^n\setminus U]\cdot\frac{\epsilon}{2}\right]\\
    &\leq\epsilon
  \end{align*}
  \dw{check}
\end{proof}

\subsection{Type Soundness}
In order to aggregate the effect of transformations we employ lists (typically denoted by $\Phi$)
of diffeomorphisms.
A list $[\phi_{(-)}^{(1)},\ldots,\phi_{(-)}^{(n)}]$ of diffeomorphisms $\parasp\times\Real\to\Real$ defines a diffeomorphism
\begin{align*}
  \rep_{(-)}:\parasp\times\Real^n&\to\Real^n\\
  (\para,[s_1,\ldots,s_n])&\mapsto\left(\phi^{(1)}_\para(s_1),\ldots,\phi^{(n)}_\para(s_n)\right)
\end{align*}
and we use concatentation notation.

We posit the following infinitary logical relation $\rel^\Phi_\tau$ between sequences of elements $\parasp\times\Real^n\to\sem\tau$ in $\VFr$ (corresponding to the smoothings) and $\parasp\times\Real^n\to\sem\tau$ in $\qbs$ (corresponding to the measurable standard semantics):
\begin{enumerate}
  \item $(f_\eta,f)\in\rel^\Phi_{\iota^{(\false,\dep)}}$ if $f_\eta\sau f$
  \item $(f_\eta,f)\in\rel^\Phi_{\iota^{(\true,\dep)}}$ if $f_\eta\sau f$, $f_\eta=g_\eta\circ\rep_{(-)}$ and $f=g\circ\rep_{(-)}$, where
  \begin{enumerate}
    \item $\rep$ is defined by $\Phi$ as above
    \item\label{it:nae} $g:\Real^n\to\Real$ is piecewise \dw{countable} analytic and non-constant
    \item\label{it:dep} on each piece \dw{???} $g$ may only depend on (transformed) $z_j$ if $s_j\in\dep$
  \end{enumerate}
  \item $(f_\eta,f)\in\rel^\Phi_{\tau_1\addtr\trt_3\to\tau_2}$ iff for all $\Phi_2$ and $(g_\eta,g)\in\rel^{\Phi\app\Phi_2}_{\tau_1}$, there exists $\Phi_3$ such that $|\Phi_3|=|\trt_3|$ and $(f_\eta\lrcomp g_\eta,f\lrcomp g)\in\rel^{\Phi\app\Phi_2\app\Phi_3}_{\tau_2}$.
\end{enumerate}
Note that \cref{it:nae} implies $f\neq 0$ a.e.\ because non-constant analytic functions vanish on negligible sets \cite{M15} and diffeomorphisms preserve negligibility.

\begin{lemma}
  \label{lem:relcond}
   If $(f_\eta,f)\in\rel^{\Phi}_{\PCFReal^{(\true,\dep)}}$ and $(g_\eta,g),(h_\eta,h)\in\rel^{\Phi}_\safet$ then \dw{notation}
   \begin{align*}
     ((\smooth\circ (-f_\eta))\cdot g_\eta+(\smooth\circ f_\eta)\cdot h_\eta ,[f(-)<0]\cdot g+[f(-)\geq 0]\cdot h)\in\rel^\Phi_\safet
   \end{align*}
\end{lemma}
\begin{proof}
  We focus on the argument for the case where $\safet$ is the annotated base type, in particular $\iota^{(\true,\dep)}$, which is most interesting; the extension to higher orders can be obtained similarly as for \cref{lem:ppredcond}.
  Clearly, \cref{it:nae,it:dep} are satisfied and u.a.u.\ convergence follows from \cref{lem:ausig,lem:aucomp}.
\end{proof}

  Intuitively, $\Phi$ describes how samples which \emph{may} have been drawn during execution are transformed
  We can add additional samples, which are ignored:
  \begin{lemma}
    \label{lem:extpred}
    Let $(f_\eta,f)\in\rel_{\Phi,\tau}$ and $\Phi'$ be a list of diffeomorphisms.
    Then $(g_\eta,g)\in\rel_{\Phi\app\Phi',\tau}$, where $g_\eta(\para,\tr\app\tr')\defeq f_\eta(\para,\tr)$ and
    $g(\para,\tr\app\tr')\defeq f(\para,\tr)$.
  \end{lemma}

\begin{restatable}[Fundamental]{lemma}{fundconv}
  \label{lem:fundconv}
If $\theta_1:\iota_1^{(\false,\emptyset)},\ldots,\theta_m:\iota_m^{(\false,\emptyset)},x_1:\tau_1,\ldots,x_\ell:\tau_\ell\mid\trt\vdash M:\tau$, $\Phi$ be a list of diffeomorphisms, $(\xi_\eta^{(1)},\xi^{(1)})\in\rel^\Phi_{\tau_1},\ldots,(\xi_\eta^{(\ell)},\xi^{(\ell)})\in\rel^\Phi_{\tau_\ell}$ then there exists a list $\Phi'$ of diffeomorphisms that $|\trt|=|\Phi'|$ and
$(\sema M\flcomp\langle\xi_\eta^{(1)},\ldots,\xi_\eta^{(\ell)}\rangle, \sem M\flcomp\langle\xi^{(1)},\ldots,\xi^{(\ell)}\rangle)\in\rel^{\Phi\app\Phi'}_\tau$, where $\flcomp$ is defined as in \cref{lem:fundpoly1}.
\end{restatable}
\begin{proof}
The claim is proven by induction on the typing judgements. We focus on the most interesting cases:
\begin{enumerate}
  \item For conditionals we exploit the inductive hypothesis and \cref{lem:relcond}.
  \item Suppose
  $\para\mid[s_j\sim\dist]\gsj\transt\dist T:\PCFReal^{(\true,\{s_j\})}$ because $T$ is diffeomorphic.
  We define
  \begin{align*}
    g(s_j)&\defeq s_j&
    \phi_\para(s)&\defeq\sem T(\para,\etr)(s,\etr)=\sema T(\para,\etr)(s,\etr)
  \end{align*}
  and therefore we can easily see that
  \begin{align*}
    \sema{\transt\dist T}=g\circ\phi_{(-)}=\sem{\transt\dist T}
  \end{align*}
  and $(\sema{\transt\dist T},\sem{\transt\dist T})\in\rel^{[\phi_{(-)}]}_{\PCFReal^{(\true,\{s_j\})}}$ follows immediately.
  \item For addition we focus on the interesting case   $\mid[]\gsj\PCFplus:\iota^{(\true,\dep_1)}\to\iota^{(\true,\dep_2)}\to\iota^{(\true,\dep_1\cup\dep_2)}$, where $\dep_1\cap\dep_2=\emptyset$. Let $\Phi$, $\Phi_1$ and $\Phi_2$ be lists of diffeomorphisms, $(f^{(1)}_\eta,f^{(1)})\in\rel^{\Phi\app\Phi_1}_{\iota^{(\true,\dep_1)}}$ and $(f^{(2)}_\eta,f^{(2)})\in\rel^{\Phi\app\Phi_1\app\Phi_2}_{\iota^{(\true,\dep_2)}}$.
  By definition there are decompositions
  \begin{align*}
    f^{(1)}_\eta&=g_\eta^{(1)}\circ\rep_{(-)}^{(1)}&
    f^{(1)}&=g^{(1)}\circ\rep_{(-)}^{(1)}&
    f^{(2)}_\eta&=g_\eta^{(2)}\circ\rep_{(-)}^{(2)}&
    f^{(2)}&=g^{(2)}\circ\rep_{(-)}^{(2)}
  \end{align*}
  \dw{agree}
  Let $\widehat {g^{(1)}_\eta}$ and $\widehat {g^{(1)}}$ be the extension of $g^{(1)}_\eta$ and $g^{(1)}$, respectively, to $\Real^{|\Phi|+|\Phi_1|+|\Phi_2|}\to\Real$. Note that
  \begin{align*}
    \sema\PCFplus\lrcomp f_\eta^{(1)}\lrcomp f_\eta^{(2)}&= (\widehat {g_\eta^{(1)}}+g_\eta^{(2)})\circ \rep_\para^{(2)}&
    \sem\PCFplus\lrcomp f^{(1)}\lrcomp f^{(2)}= (\widehat {g^{(1)}}+g^{(2)})\circ\rep_\para^{(2)}
  \end{align*}
  Clearly (using \cref{lem:extpred}), $\widehat {g^{(1)}}+g^{(2)}$ is again piecewise analytic and on each piece depends on (transformed) samples either $g^{(1)}$ or $g^{(2)}$ depends on. Furthermore, on each piece $\widehat {g^{(1)}}+g^{(2)}$ is not constant because $g^{(1)}$ and $g^{(2)}$ are not constant and depend on different variables. \dw{does this make sense?}

\end{enumerate}
\end{proof}

%% file: appendix/app-experiments.tex

\subsection{Experimental Setup}
\label{app:expt-setup}

To generate the ELBO trajectories shown in \cref{fig:estimator variance}, we separately took 1000 samples of the ELBO every 100 iterations, taking extra samples to reduce the variance in the graphs presented.
The random samples were the same across estimators, which leads to the correlation in noise seen in their trajectories.

\cref{tab:var} compares the average variance of the estimators, where the average is taken over a single optimisation trajectory.
For each estimator, we took 1000 Monte Carlo samples of the gradient every 100 iterations to compute the variance of the estimator at that iteration; we then computed the average of these variances.
Since the gradients are vectors, the variance was measured in two ways: averaging the component-wise variances and the variance of the $L_2$ norm.

We then separately benchmark each estimator by measuring how many iterations each can complete in a fixed time budget and setting the computational cost to be the reciprocal of that.
This is then used to compute a \emph{work-normalised variance} \cite{glynn1992asymptotic,Botev2017} that is taken to be the product of the computational cost and variance.
Intuitively, we divide by the relative time taken since we can reduce the variance by the same factor running the faster estimator more times.

\subsection{Models}
\label{app:models}

We include the models from \cite{LYY18}, which are as follows:
\begin{itemize}
  \item \texttt{temperature} \cite{DBLP:conf/qest/SoudjaniMN17} models a controller keeping the temperature of a room within set bounds.
  The discontinuity arises from the discrete state of the controller, being either on or off, which disrupts the continuous state representing the temperature of the room.
  Given a set of noisy measurements of the room temperature, the goal is to infer the controller state at each of 21 time steps. The model has a 41-dimensional latent variable and 80 if-statements.

  \item \texttt{textmsg} \cite{DavidsonPilon15} models daily text message rates, and the goal is to discover a change in the rate over the 74-day period of data given.
  The non-differentiability arises from the point at which the rate is modelled
  to change. The model has a 3-dimensional latent variable (the two rates and
  the point at which they change) and 37 if-statements.

  \item \texttt{influenza}  \cite{ShumwayS05} models the US influenza mortality for
  1969.
  In each month, the mortality rate depends on the dominant virus strain being of type 1 or type 2, producing a non-differentiablity for each month.
  Given the mortality data, the goal is to infer the dominant virus strain in each month. The model has a 37-dimensional latent variable and 24 if-statements.
\end{itemize}

Additionally, we introduce the following models:

\begin{itemize}
  \item \texttt{cheating} 
  \cite{DavidsonPilon15} simulates a differential privacy setting where students taking an exam are surveyed to determine the prevalence of cheating without exposing the details for any individual.
  Students are tasked to toss a coin, on heads they tell the truth (cheating or not cheating) and on tails they toss a second coin to determine their answer.
  The tossing of coins here is a source of discontinuity.
  The goal, given the proportion of students who answered yes, is to predict a posterior on the cheating rate.
  In this model there are 300 if-statements and a 301-dimensional latent space, although we only optimise over a single dimension with the other 300 being sources of randomness.

  \item \texttt{xornet} is a simple multi-layer neural network trained to compute the exclusive-or (XOR) function.
  It has a \changed[lo]{\texttt{2-4-2-1}} network architecture with two inputs and one output, and all activation functions being the Heaviside step function which is traditionally infeasible for gradient-based optimisation because of the discontinuity at $0$ and a zero gradient everywhere else.
  The model has a \changed[lo]{25-dimensional} latent space (for all the weights and biases) and 28 if-statements.
  Note that this model is not applicable to the \lyy estimator since the branch conditions are not all affine in the latent space.
\end{itemize}

\subsection{Analysis of Results}
\label{app:results}


The ELBO graph for the \texttt{temperature} model in \cref{fig:temperature-graph} shows that the \reparam estimator is biased, converging to a suboptimal value when compared to the \nested and \lyy estimators.
We can also see from the graph and the data in \cref{tab:temperature} that the \score estimator exhibits extremely high variance, and does not converge.

The \texttt{textmsg} and \texttt{influenza} ELBO graphs in \cref{fig:textmsg-graph} and \cref{fig:influenza-graph} both show all estimators converging towards roughly the same value, with \score exhibiting a larger variance.
The work-normalised variance of the \nested estimator across both model is the lowest across both variance measures.

For the \texttt{cheating} model in \cref{fig:cheating-graph}, we have another visual indicator of the bias of the \reparam gradient.
Here \nested outperforms again with the lowest work-normalised variance (ignoring that of \reparam since it is biased).

Finally, the \texttt{xornet} model shows the difficulty of training step-function based neural nets.
The \lyy estimator is not applicable here since the boundary integral has no general efficient estimator for non-affine conditionals, which is the case here.
In \cref{fig:xornet-graph}, the \reparam estimator makes no progress while other estimators manage to converge to close to $0$ ELBO, showing that they learn a network that correctly classifies all points.
In particular, the \nested estimator converges the quickest.

To summarise, the results show cases where the \reparam estimator is biased and how the \nested estimator do not have the same limitation.
Where the \lyy estimator is defined, they converge to roughly the same objective value; and the smoothing approach is generalisable to more complex models such as neural networks with non-linear boundaries.
Our proposed \nested estimator has consistently significantly lower work-normalised variance, up to 3 orders of magnitude.



\begin{table}[ht]
  \caption{Computational cost and work-normalised variances, all given as ratios
  with respect to the \textsc{Score} estimator (whose data are omitted since they would be a row of $1$s). We chose $\eta=0.15$ for \textsc{Smooth}.}
  \label{tab:var}
  \centering
 \small
  \begin{subtable}[h]{.4\textwidth}
    \caption{\texttt{temperature}}
    \label{tab:temperature}
    \centering
    \begin{tabular}{lcccr}
      \toprule
      Estimator & Cost & $\textrm{Avg}(V(.))$ & $V(\|.\|_2)$ \\
      \midrule
  \textsc{Smooth}  &  1.62e+00  &  3.17e-10  &  2.09e-09    \\
  \textsc{Reparam}  &     1.28e+00  &  1.48e-08 &   2.01e-08   \\
  \textsc{Lyy18}  &       9.12e+00 &   1.22e-06 &   4.76e-05 \\
      \bottomrule
    \end{tabular}
  \end{subtable}
\qquad\qquad
  \begin{subtable}[h]{.4\textwidth}
    \caption{\texttt{textmsg}}
    \label{tab:textmsg}
    \centering
    \begin{tabular}{lcccr}
      \toprule
      Estimator & Cost & $\textrm{Avg}(V(.))$ & $V(\|.\|_2)$ \\
      \midrule
      \textsc{Smooth} & 2.00e+00 & 2.29e-02 & 3.79e-02 \\
      \textsc{Reparam} & 1.18e+00 & 1.43e-02 & 2.29e-02 \\
      \textsc{Lyy18} & 4.00e+00 & 5.76e-02 & 8.46e-02 \\
      \bottomrule
    \end{tabular}
  \end{subtable}
\\
 \vspace*{10pt}
  \begin{subtable}[h]{.4\textwidth}
    \caption{\texttt{influenza}}
    \label{tab:influenza}
    \centering
    \begin{tabular}{lcccr}
      \toprule
      Estimator & Cost & $\textrm{Avg}(V(.))$ & $V(\|.\|_2)$ \\
      \midrule
    \textsc{Smooth}  &      1.47e+00  &  9.15e-03  &  4.58e-03    \\
    \textsc{Reparam}  &     1.17e+00  &  7.45e-03  &  3.68e-03    \\
    \textsc{Lyy18}  &       8.30e+00  &  5.88e-02  &  2.91e-02 \\
      \bottomrule
    \end{tabular}
  \end{subtable}
\qquad\qquad
  \begin{subtable}[h]{.4\textwidth}
    \caption{\texttt{cheating}}
    \label{tab:cheating}
    \centering
    \begin{tabular}{lcccr}
      \toprule
      Estimator & Cost & $\textrm{Avg}(V(.))$ & $V(\|.\|_2)$ \\
      \midrule
      \textsc{Smooth} &      1.59e+00  &  3.64e-03 &   5.94e-03    \\
      \textsc{Reparam} &     9.66e-01  &  6.47e-19 &   1.74e-18    \\
      \textsc{Lyy18} &        2.51e+00 &   5.39e-02 &   1.34e-01\\
      \bottomrule
    \end{tabular}
  \end{subtable}
\\
\vspace*{10pt}
  \begin{subtable}[h]{.5\textwidth}
    \caption{\texttt{xornet}}
    \label{tab:xornet}
    \centering
    \begin{tabular}{lcccr}
      \toprule
      Estimator & Cost & $\textrm{Avg}(V(.))$ & $V(\|.\|_2)$ \\
      \midrule
      \textsc{Smooth} &      1.66e+00  &  9.57e-03 &   4.46e-02    \\
      \textsc{Reparam} &     3.51e-01  &  7.55e-09 &   2.37e-09 \\
      \bottomrule
    \end{tabular}
  \end{subtable}
\end{table}

%% file: main.bbl
\begin{thebibliography}{10}
\providecommand{\url}[1]{\texttt{#1}}
\providecommand{\urlprefix}{URL }
\providecommand{\doi}[1]{https://doi.org/#1}

\bibitem{Aumann61}
Aumann, R.J.: Borel structures for function spaces. Illinois Journal of
  Mathematics  \textbf{5} (1961)

\bibitem{B15}
Bertsekas, D.: Convex optimization algorithms. Athena Scientific (2015)

\bibitem{BT00}
Bertsekas, D.P., Tsitsiklis, J.N.: Gradient convergence in gradient methods
  with errors. {SIAM} J. Optim.  \textbf{10}(3),  627--642 (2000)

\bibitem{pyro}
Bingham, E., Chen, J.P., Jankowiak, M., Obermeyer, F., Pradhan, N., Karaletsos,
  T., Singh, R., Szerlip, P.A., Horsfall, P., Goodman, N.D.: Pyro: Deep
  universal probabilistic programming. J. Mach. Learn. Res.  \textbf{20},
  28:1--28:6 (2019)

\bibitem{B07}
Bishop, C.M.: Pattern recognition and machine learning, 5th Edition.
  Information science and statistics, Springer (2007)

\bibitem{BKM17}
Blei, D.M., Kucukelbir, A., McAuliffe, J.D.: Variational inference: A review
  for statisticians. Journal of the American Statistical Association
  \textbf{112}(518),  859--877 (2017). \doi{10.1080/01621459.2017.1285773}

\bibitem{BLGS16}
Borgstr{\"{o}}m, J., Lago, U.D., Gordon, A.D., Szymczak, M.: A lambda-calculus
  foundation for universal probabilistic programming. In: Proceedings of the
  21st {ACM} {SIGPLAN} International Conference on Functional Programming,
  {ICFP} 2016, Nara, Japan, September 18-22, 2016. pp. 33--46 (2016)

\bibitem{Botev2017}
Botev, Z., Ridder, A.: {Variance Reduction}. In: Wiley StatsRef: Statistics
  Reference Online, pp.~1--6 (2017)

\bibitem{DBLP:conf/pldi/Cusumano-Towner19}
Cusumano{-}Towner, M.F., Saad, F.A., Lew, A.K., Mansinghka, V.K.: Gen: a
  general-purpose probabilistic programming system with programmable inference.
  In: McKinley, K.S., Fisher, K. (eds.) Proceedings of the 40th {ACM} {SIGPLAN}
  Conference on Programming Language Design and Implementation, {PLDI} 2019,
  Phoenix, AZ, USA, June 22-26, 2019. pp. 221--236. {ACM} (2019).
  \doi{10.1145/3314221.3314642}, \url{https://doi.org/10.1145/3314221.3314642}

\bibitem{DK20}
Dahlqvist, F., Kozen, D.: Semantics of higher-order probabilistic programs with
  conditioning. Proc. {ACM} Program. Lang.  \textbf{4}({POPL}),  57:1--57:29
  (2020)

\bibitem{DavidsonPilon15}
Davidson-Pilon, C.: Bayesian Methods for Hackers: Probabilistic Programming and
  Bayesian Inference. Addison-Wesley Professional (2015)

\bibitem{ETP14}
Ehrhard, T., Tasson, C., Pagani, M.: Probabilistic coherence spaces are fully
  abstract for probabilistic {PCF}. In: The 41st Annual {ACM} {SIGPLAN-SIGACT}
  Symposium on Principles of Programming Languages, {POPL} '14, San Diego, CA,
  USA, January 20-21, 2014. pp. 309--320 (2014)

\bibitem{FK88}
Frölicher, A., Kriegl, A.: Linear Spaces and Differentiation Theory.
  Interscience, J. Wiley and Son, New York (1988)

\bibitem{glynn1992asymptotic}
Glynn, P.W., Whitt, W.: The asymptotic efficiency of simulation estimators.
  Operations research  \textbf{40}(3),  505--520 (1992)

\bibitem{Heunen2017c}
Heunen, C., Kammar, O., Staton, S., Yang, H.: {A convenient category for
  higher-order probability theory}. Proc. Symposium Logic in Computer Science
  (2017)

\bibitem{HKSY17}
Heunen, C., Kammar, O., Staton, S., Yang, H.: A convenient category for
  higher-order probability theory. In: 32nd Annual {ACM/IEEE} Symposium on
  Logic in Computer Science, {LICS} 2017, Reykjavik, Iceland, June 20-23, 2017.
  pp. 1--12 (2017)

\bibitem{HNRS15}
Hur, C., Nori, A.V., Rajamani, S.K., Samuel, S.: A provably correct sampler for
  probabilistic programs. In: 35th {IARCS} Annual Conference on Foundation of
  Software Technology and Theoretical Computer Science, {FSTTCS} 2015, December
  16-18, 2015, Bangalore, India. pp. 475--488 (2015)

\bibitem{JGP17}
Jang, E., Gu, S., Poole, B.: Categorical reparameterization with
  gumbel-softmax. In: 5th International Conference on Learning Representations,
  {ICLR} 2017, Toulon, France, April 24-26, 2017, Conference Track Proceedings
  (2017)

\bibitem{DBLP:journals/corr/KingmaB14}
Kingma, D.P., Ba, J.: Adam: {A} method for stochastic optimization. In: Bengio,
  Y., LeCun, Y. (eds.) 3rd International Conference on Learning
  Representations, {ICLR} 2015, San Diego, CA, USA, May 7-9, 2015, Conference
  Track Proceedings (2015)

\bibitem{DBLP:journals/corr/KingmaW13}
Kingma, D.P., Welling, M.: Auto-encoding variational bayes. In: Bengio, Y.,
  LeCun, Y. (eds.) 2nd International Conference on Learning Representations,
  {ICLR} 2014, Banff, AB, Canada, April 14-16, 2014, Conference Track
  Proceedings (2014)

\bibitem{K13}
Klenke, A.: Probability Theory: A Comprehensive Course. Universitext, Springer
  London (2014)

\bibitem{LYRY19}
Lee, W., Yu, H., Rival, X., Yang, H.: Towards verified stochastic variational
  inference for probabilistic programs. {PACMPL}  \textbf{4}({POPL}) (2020)

\bibitem{LYY18}
Lee, W., Yu, H., Yang, H.: Reparameterization gradient for non-differentiable
  models. In: Advances in Neural Information Processing Systems 31: Annual
  Conference on Neural Information Processing Systems 2018, NeurIPS 2018, 3-8
  December 2018, Montr{\'{e}}al, Canada. pp. 5558--5568 (2018)

\bibitem{DBLP:journals/pacmpl/LewCSCM20}
Lew, A.K., Cusumano{-}Towner, M.F., Sherman, B., Carbin, M., Mansinghka, V.K.:
  Trace types and denotational semantics for sound programmable inference in
  probabilistic languages. Proc. {ACM} Program. Lang.  \textbf{4}({POPL}),
  19:1--19:32 (2020)

\bibitem{MMT17}
Maddison, C.J., Mnih, A., Teh, Y.W.: The concrete distribution: {A} continuous
  relaxation of discrete random variables. In: 5th International Conference on
  Learning Representations, {ICLR} 2017, Toulon, France, April 24-26, 2017,
  Conference Track Proceedings (2017)

\bibitem{MOPW}
Mak, C., Ong, C.L., Paquet, H., Wagner, D.: Densities of almost surely
  terminating probabilistic programs are differentiable almost everywhere. In:
  Yoshida, N. (ed.) Programming Languages and Systems - 30th European Symposium
  on Programming, {ESOP} 2021, Held as Part of the European Joint Conferences
  on Theory and Practice of Software, {ETAPS} 2021, Luxembourg City,
  Luxembourg, March 27 - April 1, 2021, Proceedings. Lecture Notes in Computer
  Science, vol. 12648, pp. 432--461. Springer (2021)

\bibitem{DBLP:conf/icml/MnihG14}
Minh, A., Gregor, K.: Neural variational inference and learning in belief
  networks. In: Proceedings of the 31th International Conference on Machine
  Learning, {ICML} 2014, Beijing, China, 21-26 June 2014. {JMLR} Workshop and
  Conference Proceedings, vol.~32, pp. 1791--1799. JMLR.org (2014)

\bibitem{M15}
Mityagin, B.: The zero set of a real analytic function (2015)

\bibitem{M99}
Munkres, J.R.: Topology. Prentice Hall, New Delhi,, 2nd. edn. (1999)

\bibitem{Murphy2012}
Murphy, K.P.: Machine Learning: A Probabilististic Perspective. MIT Press
  (2012)

\bibitem{RGB14}
Ranganath, R., Gerrish, S., Blei, D.M.: Black box variational inference. In:
  Proceedings of the Seventeenth International Conference on Artificial
  Intelligence and Statistics, {AISTATS} 2014, Reykjavik, Iceland, April 22-25,
  2014. pp. 814--822 (2014)

\bibitem{DBLP:conf/icml/RezendeMW14}
Rezende, D.J., Mohamed, S., Wierstra, D.: Stochastic backpropagation and
  approximate inference in deep generative models. In: Proceedings of the 31th
  International Conference on Machine Learning, {ICML} 2014, Beijing, China,
  21-26 June 2014. {JMLR} Workshop and Conference Proceedings, vol.~32, pp.
  1278--1286. JMLR.org (2014)

\bibitem{ShumwayS05}
Shumway, R.H., Stoffer, D.S.: Time Series Analysis and Its Applications.
  Springer Texts in Statistics, Springer-Verlag (2005)

\bibitem{DBLP:conf/qest/SoudjaniMN17}
Soudjani, S.E.Z., Majumdar, R., Nagapetyan, T.: Multilevel monte carlo method
  for statistical model checking of hybrid systems. In: Bertrand, N.,
  Bortolussi, L. (eds.) Quantitative Evaluation of Systems - 14th International
  Conference, {QEST} 2017, Berlin, Germany, September 5-7, 2017, Proceedings.
  Lecture Notes in Computer Science, vol. 10503, pp. 351--367. Springer (2017)

\bibitem{S11}
Stacey, A.: Comparative smootheology. Theory and Applications of Categories
  \textbf{25}(4),  64--117 (2011)

\bibitem{S17}
Staton, S.: Commutative semantics for probabilistic programming. In:
  Programming Languages and Systems - 26th European Symposium on Programming,
  {ESOP} 2017, Held as Part of the European Joint Conferences on Theory and
  Practice of Software, {ETAPS} 2017, Uppsala, Sweden, April 22-29, 2017,
  Proceedings. pp. 855--879 (2017)

\bibitem{SYWHK16}
Staton, S., Yang, H., Wood, F.D., Heunen, C., Kammar, O.: Semantics for
  probabilistic programming: higher-order functions, continuous distributions,
  and soft constraints. In: Proceedings of the 31st Annual {ACM/IEEE} Symposium
  on Logic in Computer Science, {LICS} '16, New York, NY, USA, July 5-8, 2016.
  pp. 525--534 (2016)

\bibitem{DBLP:conf/icml/TitsiasL14}
Titsias, M.K., L{\'{a}}zaro{-}Gredilla, M.: Doubly stochastic variational bayes
  for non-conjugate inference. In: Proceedings of the 31th International
  Conference on Machine Learning, {ICML} 2014, Beijing, China, 21-26 June 2014.
  pp. 1971--1979 (2014)

\bibitem{VKS19}
V{\'{a}}k{\'{a}}r, M., Kammar, O., Staton, S.: A domain theory for statistical
  probabilistic programming. {PACMPL}  \textbf{3}({POPL}),  36:1--36:29 (2019)

\bibitem{WW13}
Wingate, D., Weber, T.: Automated variational inference in probabilistic
  programming. CoRR  \textbf{abs/1301.1299} (2013)

\bibitem{Z81}
Zang, I.: Discontinuous optimization by smoothing. Mathematics of Operations
  Research  \textbf{6}(1),  140--152 (1981)

\bibitem{Zhang2019}
Zhang, C., Butepage, J., Kjellstrom, H., Mandt, S.: {Advances in Variational
  Inference}. IEEE Trans. Pattern Anal. Mach. Intell.  \textbf{41}(8),
  2008--2026 (2019)

\end{thebibliography}
